\documentclass{article}

\usepackage{fullpage}
\usepackage{amsmath}
\usepackage{amssymb}

\usepackage{amsthm}
\usepackage{stmaryrd}
\usepackage{nameref,hyperref,cleveref}
\usepackage[normalem]{ulem}
\usepackage[utf8]{inputenc} 
\usepackage[T1]{fontenc}
\usepackage{shuffle}

\usepackage{pxfonts}
\usepackage{proof}

\usepackage{graphicx}
\usepackage[all,2cell]{xy}
\UseTwocells

\usepackage{enumitem}
\setlist[enumerate]{label=\arabic*)}

\newtheorem{theorem}{Theorem}[section]
\crefname{theorem}{thm.}{theorems}
\newtheorem{lemma}[theorem]{Lemma}
\newtheorem{corollary}[theorem]{Corollary}
\crefname{corollary}{cor.}{corollaries}

\newtheorem{proposition}[theorem]{Proposition}
\crefname{proposition}{prop.}{propositions}

\newtheorem{remark}[theorem]{Remark}

\theoremstyle{definition}
\newtheorem{definition}[theorem]{Definition}
\crefname{definition}{defn.}{definitions}

\newtheorem{example}{Example}

\newcommand\definand[1]{{\bf #1}}
\newcommand\defeq{\stackrel{\text{\tiny def}}{=}}

\def\cc#1{{\mathcal{#1}}}

\def\b{\mathbf{b}}
\def\t{\mathbf{t}}
\def\rs{\mathbf{t}}
\newcommand\dom{\mathop{\rm dom}}
\newcommand\cod{\mathop{\rm cod}}
\newcommand\id{{\rm id}}
\newcommand\op{{\rm op}}

\newcommand\Set{ {\mathbf{Set} } }
\newcommand\Cat{ {\mathbf{Cat} } }

\newcommand\Psh{\mathbf{Psh}}

\newcommand\refs{\sqsubset}
\newcommand\seq[1]{\underset{#1}\Longrightarrow}

\newcommand\nseq{\Longrightarrow}
\newcommand\seqd[2]{\overset{\displaystyle #1}{\underset{#2}{\vphantom{X}\Longrightarrow\vphantom{X}}}}
\newcommand\viso{\equiv}
\newcommand\visod[1]{\overset{\displaystyle #1}{\vphantom{X}\viso}}

\let\G=\Gamma
\let\D=\Delta
\newcommand\vd{\vdash}

\newcommand\conv{\sim}
\newcommand\deq{=}

\newcommand\pull[1]{{#1}^*\,}
\newcommand\push[1]{{#1}\,}

\newcommand\ImpL{\multimap}
\newcommand\ImpR{\mathbin{\text{\reflectbox{$\multimap$}}}}

\newcommand\impL[3][]{{#3} \mathbin{\ImpL}_{#1} {#2}}
\newcommand\impR[3][]{{#2} \mathbin{{}_{#1}\ImpR} {#3}}
\newcommand\resL[2][\bot]{{#2} \mathbin{\setminus} {#1}}
\newcommand\resR[2][\bot]{{#1} \mathbin{/} {#2}}

\newcommand\lc[1]{\lambda[{#1}]}
\newcommand\rc[1]{\rho[{#1}]}
\newcommand\plugL{leval}
\newcommand\plugR{reval}

\providecommand\oast{\mathop{\circledast}}
\newcommand\mul{\mathbin{\bullet}}

\newcommand\perpL[2][\bot]{{#2}^{\bot#1}}
\newcommand\perpR[2][\bot]{{}^{\bot#1}{#2}}

\newcommand\negoL[2][]{\perpL[]{#2}}
\newcommand\negoR[2][]{\perpR[]{#2}}

\newcommand\triple[3]{\{#1\}#2\{#3\}}

\newcommand\set[1]{\{\,#1\,\}}

\newcommand\upc{\mathbf{u}}
\newcommand\obc{\mathbf{s}}
\newcommand\commacat[2]{{#1}\downarrow{#2}}

\newcommand\pt[2][]{#2^{+#1}}
\newcommand\ptfun[1]{+#1}
\newcommand\nt[2][]{#2^{-#1}}
\newcommand\ntfun[1]{-#1}

\newcommand\pgro[1]{\partial^+{#1}}
\newcommand\ngro[1]{\partial^-{#1}}

\newcommand\str{m}
\newcommand\costr{a}
\newcommand\costrr{a'}

\newcommand\Typ[1]{\cc{T}^{\sharp{#1}}}
\newcommand\Der[1]{\cc{D}^{\sharp{#1}}}

\newcommand\cut[1][]{{\textstyle\oast_{#1}}}

\newcommand\unit{\iota}
\newcommand\coun{o}

\def\hid#1{}

\newcommand\simplex{\Delta}

\newcommand\fin[1]{{\bf #1}}
\newcommand\linFm{\cc{F}}
\newcommand\linCtx{{\cc{W}}}
\newcommand\FMon[1]{\mathbf{M}[#1]}
\newcommand\FSMon[1]{{\mathbf{SM}[#1]}}
\newcommand\Fin{\mathbf{Fin}}

\newcommand\fctx[2]{#1[#2]}

\newcommand\Diagrams[1]{\Cat \sslash #1}

\newcommand\pairing{p}

\usepackage{eso-pic}

\title{An Isbell Duality Theorem for Type Refinement Systems}

\author{Paul-André Melliès and Noam Zeilberger}

\begin{document}

\maketitle

\begin{abstract}
Any refinement system (= functor) has a fully faithful representation in the refinement system of presheaves, by interpreting types as relative slice categories, and refinement types as presheaves over those categories.
Motivated by an analogy between side effects in programming and \emph{context effects} in linear logic, we study logical aspects of this ``positive'' (covariant) representation, as well as of an associated ``negative'' (contravariant) representation.
We establish several preservation properties for these representations, including a generalization of Day's embedding theorem for monoidal closed categories.
Then we establish that the positive and negative representations satisfy an Isbell-style duality.
As corollaries, we derive two different formulas for the positive representation of a pushforward (inspired by the classical negative translations of proof theory), which express it either as the dual of a pullback of a dual, or as the double dual of a pushforward.
Besides explaining how these constructions on refinement systems generalize familiar category-theoretic ones (by viewing categories as special refinement systems), our main running examples involve representations of Hoare Logic and linear sequent calculus.
\end{abstract}



\tableofcontents

\section{Introduction}
\label{sec:intro}

This paper continues the study of type systems from the perspective outlined in \cite{mz15popl}.
There, we suggested that it is useful to view a type system as a functor from a category of typing derivations to a category of underlying terms, and that this can even serve as a working \emph{definition} of ``type system'' (or what we call a \emph{refinement system}), as being (in the most general case) simply an arbitrary functor.
\begin{definition}
A \definand{(type) refinement system} is a functor $\rs : \cc{D} \to \cc{T}$.
\end{definition}
\begin{definition}\label{defn:refine}
We say that an object $P \in \cc{D}$ \definand{refines} an object $A \in \cc{T}$ (notated $P \refs A$) if $\rs(P) = A$.
\end{definition}
\begin{definition}\label{defn:judge}
A \definand{typing judgment} is a triple $(P,c,Q)$, where $c$ is a morphism of $\cc{T}$ such that $P \refs \dom(c)$ and $Q \refs \cod(c)$ (notated $P \seq{c} Q$). In the special case where $P$ and $Q$  refine the same object $P,Q \refs A$ and $c$ is the identity morphism $c = \id_A$, the typing judgment $(P,c,Q)$ is also called a \definand{subtyping judgment} (notated $P \nseq Q$).
\end{definition}
\begin{definition}\label{defn:derive}
A \definand{derivation} of a (sub)typing judgment $(P,c,Q)$ is a morphism $\alpha : P \to Q$ in $\cc{D}$ such that $\rs(\alpha) = c$ (notated $\deduce{P \seq{c} Q}{\alpha}$).
\end{definition}
\noindent
It is useful to think of any category $\cc{C}$ as a trivial example of a refinement system in at least two ways -- either as the identity functor $\id_{\cc{C}} : \cc{C} \to \cc{C}$ or as the terminal functor $!_{\cc{C}} : \cc{C} \to 1$ -- and in the sequel we will describe several general constructions on refinement systems, which reduce to classical constructions on categories seen as such degenerate refinement systems.
Another important (and more guiding) example of a refinement system is Hoare Logic \cite{hoare69}:
\begin{itemize}
\item Take $\cc{T}$ to be a category with a single object $W$ representing the state space and morphisms $c : W \to W$ corresponding to state transformers.
\item Take $\cc{D}$ to be a category whose objects $P,Q \in \cc{D}$ are predicates over the state space $W$ and whose morphisms $(c,\alpha) : P \to Q$ are commands $c : W \to W$ equipped with a verification $\alpha$ that $c$ will take any state satisfying $P$ to a state satisfying $Q$.
\item Take $\rs : \cc{D} \to \cc{T}$ to be the evident forgetful functor.
\end{itemize}
In this case, a typing judgment is nothing but a \emph{Hoare triple} $\triple{P}c{Q}$, and what the example highlights is that a typing judgment can describe not just a logical entailment but also a \emph{side effect} (here the transformation $c$ upon the state).

In fact, one of our original motivations for studying this framework came from apparent connections between side effects and linear logic, and in particular its \emph{proof theory} \cite{girard87linear,andreoli92}.
Let us illustrate this idea by considering the right-rule for multiplicative conjunction (``tensor'') in intuitionistic linear logic:
$$
\infer[{\otimes}R]{\G,\D \vd A \otimes B}{\G \vd A & \D \vd B}
$$
Following the tradition of \cite{gentzen1935thesis}, it is common to call $A \otimes B$ the \emph{principal formula} of the ${\otimes}R$ rule, and $A$ and $B$ its \emph{side formulas}.
The letters $\Gamma$ and $\Delta$ then stand for arbitrary sequences of formulas -- often called \emph{contexts} -- which are carried through from the premises into the conclusion.
Now, one can try to internalize the fact that the inference rule is \emph{parametric} in $\G$ and $\D$ by first organizing contexts into some category $\linCtx$.
Assuming a reasonable definition of morphism (between contexts) in $\linCtx$, any formula $C$ then induces a presheaf $\pt{C} : \linCtx^\op \to \Set$ by considering all the proofs of $C$ in a given context:
$$
\pt{C} = \G \mapsto \set{\pi \mid \deduce{\G \vd C}{\pi}}
$$
For example, we could take $\linCtx$ as a category whose objects are lists (or multisets) of linear logic formulas and whose morphisms are \emph{linear substitutions}, i.e., where a morphism $\D \to \G$ is given by a list of proofs
$$
\deduce{\D_1 \vdash A_1}{\pi_1}
\quad\cdots\quad
\deduce{\D_n \vdash A_n}{\pi_n}
$$
such that $\D = \D_1,\dots,\D_n$ and $\G = A_1,\dots,A_n$.
With that definition of $\linCtx$, the functorial action of $\pt{C}$ is just to perform a multicut: given a proof $\pi$ of $\G \vdash C$ and a linear substitution $\sigma : \D \to \G$, one obtains a proof $\pt{C}(\sigma)(\pi)$ of $\D \vdash C$ by cutting the proofs $\sigma = (\pi_1,\dots,\pi_n)$ for the assumptions $\G = A_1,\dots,A_n$ in $\pi$.

Next, noting that for any pair of presheaves
$$\phi_1 : \cc{C}_1^\op \to \Set \quad\text{and}\quad \phi_2 : \cc{C}_2^\op \to \Set$$
one can construct their \emph{external tensor product} as the presheaf $\phi_1 \mul \phi_2 : (\cc{C}_1 \times \cc{C}_2)^\op \to \Set$ defined by
$$(\phi_1\mul \phi_2)(x_1,x_2) = \phi_1(x_1) \times \phi_2(x_2),$$
we might hope to represent the fact that the ${\otimes}R$ rule is parametric in $\G$ and $\D$ by interpreting ${\otimes}R$ as a \emph{natural transformation} from the external tensor product $\pt A \mathbin{\mul} \pt B$ to the ``internal'' tensor product $\pt{(A\otimes B)}$.
The difficulty is that this is not well-typed!
The point is that $\pt A \mathbin{\mul} \pt B$ is a presheaf over the product category $\linCtx \times \linCtx$, whereas $\pt{(A\otimes B)}$ is a presheaf over $\linCtx$, and so, literally interpreted, it does not make sense to speak of natural transformations between them.
What is missing is that the ${\otimes}R$ rule also has an implicit ``context effect'', namely the operation of concatenating (or taking a multiset union of) $\Gamma$ and $\Delta$.
If we make this operation explicit as a \emph{functor}
$$m : \linCtx \times \linCtx \to \linCtx$$
then we can literally interpret the ${\otimes}R$ rule as an honest natural transformation: not directly between the two presheaves $\pt A \mul \pt B$ to $\pt{(A\otimes B)}$, but rather from $\pt A \mul \pt B$ to the presheaf $\pt{(A\otimes B)}$ \emph{precomposed} with the functor $m$.

What this example exposes is the danger of limiting one's attention to a single presheaf category, and it suggests taking an alternative approach: to use the language of type refinement to speak directly about presheaves living in different presheaf categories.
Concretely, there is a refinement system defined as the forgetful functor $\upc : \Psh \to \Cat$, which sends a pair $(\cc{C},\phi)$ of a category $\cc{C}$ equipped with a presheaf $\phi : \cc{C}^\op \to \Set$ to the underlying category $\cc{C}$.
Our tentative interpretation of the ${\otimes}R$ rule of linear logic as a ``natural transformation with side effects'' can now be given a concise formulation, simply stating that ${\otimes}R$ can be interpreted as a derivation of the typing judgment
$$\pt A \mul \pt B \seq{m} \pt{(A\otimes B)}$$
in the refinement system $\upc: \Psh \to \Cat$.

Moreover, it turns out that this presheaf interpretation of the sequent calculus of linear logic may be vastly generalized: in fact, \emph{any refinement system} $\rs : \cc{D} \to \cc{T}$ can be given a presheaf interpretation, as a morphism of refinement systems $\rs \to \upc$ which is fully faithful in an appropriate sense.
The idea of representing logical formulas as presheaves over context categories was one of our original motivations for studying the notion of type refinement, and we believe that this embedding theorem sending any refinement system into $\upc: \Psh \to \Cat$ justifies that point of view.
After presenting some background in \Cref{sec:prelim}, we describe the ``positive'' representation $\pt{(-)} : \rs \to \upc$ of a refinement system together with an associated ``negative'' representation $\nt{(-)} : \rs^\op \to \upc$ in \Cref{sec:representing}.
We also establish there a few basic properties of these representations and consider various examples.
Then, in \Cref{sec:duality} we show that the two presheaf representations $\pt{P}$ and $\nt{P}$ of a refinement type $P\refs A$ satisfy a form of duality generalizing \emph{Isbell duality} (i.e., the duality between the covariant and contravariant representable presheaves associated to an object of a category under the Yoneda embedding).
Finally, by combining this duality theorem with preservation properties of the two presheaf representations, we show that the positive representation $\pt{(\push{c}{P})}$ of a pushforward of a refinement $P\refs A$ (in $\rs:\cc{D}\to\cc{T}$) along a morphism $c:A\to B$ can be explicitly computed (in $\upc:\Psh\to\Cat$) using either of two ``negative translation''-style formulas, which express $\pt{(\push{c}{P})}$ both as the dual of a pullback of a dual and as the double dual of a pushforward.

\section{Preliminaries}
\label{sec:prelim}

\subsection{Basic conventions and definitions}
\label{sec:basic}

We recall some conventions from \cite{mz15popl} for working with functors as type refinement systems.
Given a fixed functor $\rs : \cc{D} \to \cc{T}$, we refer to the objects of $\cc{T}$ as \emph{types}, to the morphisms of $\cc{T}$ as \emph{terms}, and to the objects of $\cc{D}$ as \emph{refinement types} (or \emph{refinements} for short).
Since these notions are relative to a functor $\rs$, to avoid ambiguity one can speak of $\rs$-types, $\rs$-refinements, and so on.
We say that a judgment $(P,c,Q)$ is \emph{valid} in a given refinement system $\rs$ if it has a derivation in the sense of \Cref{defn:derive}, i.e., there exists a morphism $\alpha:P\to Q$ in $\cc{D}$ which is mapped to $c$ by the functor $\rs$.
More generally, we say that a \emph{typing rule} is valid when there is an operation transforming derivations of the premises into a derivation of the conclusion.
For example, the rule
$$
\infer[;]{P \seq{c;d} R}{P \seq{c} Q & Q \seq{d} R}
$$
is valid for any refinement system as an immediate consequence of functoriality: given a morphism $\alpha : P \to Q$ such that $\rs(\alpha) = c$ and a morphism $\beta : Q \to R$ such that $\rs(\beta) = d$, there is a morphism $(\alpha;\beta) : P \to R$ and moreover $\rs(\alpha;\beta) = (\rs(\alpha);\rs(\beta)) = (c;d)$.

We consider $\rs$-typing judgments modulo equality of terms (i.e., equality of morphisms in $\cc{T}$), but often we mark applications of an equality by an explicit conversion step (which can be seen as admitting the possibility that $\cc{T}$ is a higher-dimensional category, although we will not pursue that idea rigorously here).
For example, the rule of ``covariant subsumption'' of subtyping
(also called ``post-strengthening'' in Hoare Logic)
$$
\infer{P \seq{c} R}{P \seq{c} Q & Q \nseq R}
$$
can be derived from the composition typing rule ($;$) just above by
$$
\infer[\conv]{P \seq{c} R}{
\infer[;]{P \seq{c;\id} R}{
 P \seq{c} Q & Q \nseq R}}
$$
where at $\conv$ we have applied the axiom $c = (c;\id)$ which is valid in any category.

The notions of a \emph{cartesian morphism} and of a \emph{fibration of categories} \cite{borceux2} may be naturally expressed in the language of refinement systems by first defining a \definand{pullback of $Q$ along $c$} as a refinement $\pull{c}Q$
$$
\infer{\pull{c}Q \refs A}{c : A \to B & Q \refs B}
$$
equipped with a pair of typing rules
$$
\infer[L\pull{c}]{\pull{c}Q \seq{c} Q}{}
\quad
\infer[R\pull{c}]{P \seq{d} \pull{c}Q}{P \seq{d;c} Q}
$$
satisfying equations
$$
\small
\infer[;]{P \seq{d;c} Q}{
 \infer[R\pull{c}]{P \seq{d} \pull{c}Q}{P \seqd{\beta}{d;c} Q} &
 \infer[L\pull{c}]{\pull{c}Q \seq{c} Q}{}
}
\ \ \deq\ \ 
P \seqd{\beta}{d;c} Q
$$
and
$$
\small
P \seqd{\eta}{d} \pull{c}Q
\ \ \deq\ \ 
\infer[R\pull{c}]{P \seq{d} \pull{c}Q}
{\infer[;]{P \seq{d;c} Q}
{P \seqd{\eta}{d} \pull{c}Q & \infer[L\pull{c}]{\pull{c}Q \seq{c} Q}{}}}
$$
Dually, a \definand{pushforward of $P$ along $c$} is defined as a refinement $\push{c}P$
$$
\infer{\push{c}P \refs B}{P \refs A & c : A \to B}
$$
equipped with a pair of typing rules
$$
\infer[Lc]{\push{c}P \seq{d} Q}{P \seq{c;d} Q}
\qquad
\infer[Rc]{P \seq{c} \push{c}P}{}
$$
satisfying a similar pair of equations.
Note that pullbacks and pushforwards are always determined up to \definand{vertical isomorphism}, where we say that two refinements $P,Q \refs A$ of a common type are vertically isomorphic (written $P \viso Q$) when there exists a pair of \emph{subtyping} derivations
$$
\deduce{P \nseq Q}\alpha
\qquad
\deduce{Q \nseq P}\beta
$$
which compose to the identities on $P$ and $Q$. 
We record the following type-theoretic transcriptions
of standard facts in the categorical literature:
\begin{proposition}\label{prop:pullpushprops} Whenever the corresponding pullbacks and/or pushforwards exist:
\begin{enumerate}
\item the following subtyping rules are valid:
$$
\infer{\pull{c}Q_1 \nseq \pull{c}Q_2}{Q_1 \nseq Q_2}
\qquad
\infer{\push{c}P_1 \nseq \push{c}P_2}{P_1 \nseq P_2}
$$
\item we have vertical isomorphisms
$$
\pull{(d;c)}Q \viso \pull{d}\pull{c}Q
\qquad
\push{(c;d)}P \viso \push{d}\push{c}P
\qquad
\pull{\id}Q \viso Q
\qquad
\push{\id}P \viso P
$$
\end{enumerate}
\end{proposition}
\noindent
A functor $\rs : \cc{D} \to \cc{T}$ is a \definand{fibration} (respectively~\definand{opfibration}) if and only if a pullback $\pull{c}Q$ (pushforward $\push{c}P$) exists for all compatible $c$ and $Q$ ($c$ and $P$).
The definition of a fibration (originally due to Grothendieck) is to a large extent motivated by the fact that there is an equivalence between fibrations $\cc{D} \to \cc{T}$ and (pseudo)functors $\cc{T}^\op \to \Cat$, and similarly between opfibrations $\cc{D} \to \cc{T}$ and (pseudo)functors $\cc{T} \to \Cat$.
The reader will observe that one direction of these equivalences
is contained in \Cref{prop:pullpushprops}.
As the definitions make plain, though, it is possible to speak of \emph{specific} pullbacks and pushforwards, even if $\rs$ is not necessarily a fibration and/or opfibration.

Recall that a category is \emph{monoidal} if it is equipped with a tensor product and unit operation
$$
\mul : \cc{C} \times \cc{C} \to \cc{C}
\qquad
I : 1 \to \cc{C}
$$
which are associative and unital up to coherent isomorphism, and that it is \emph{closed} if it is additionally equipped with left and right residuation operations
$$
\setminus : \cc{C}^\op  \times \cc{C} \to \cc{C}
\qquad
/ : \cc{C}  \times \cc{C}^\op \to \cc{C}
$$
which are right adjoint to tensor product in each component:
$$
\cc{C}(Y, \resL[Z]{X}) \cong 
\cc{C}(X \mul Y, Z) \cong
\cc{C}(X, \resR[Z]{Y})
$$
A \definand{(closed) monoidal refinement system} is a refinement system $\rs : \cc{D} \to \cc{T}$ such that $\cc{D}$ and $\cc{T}$ are (closed) monoidal, and $\rs$ strictly preserves tensor products (and residuals) and the unit.
By our conventions, a monoidal refinement system thus admits the following refinement rules and typing rules
$$
\small
\infer{P_1 \mul P_2 \refs A_1 \mul A_2}{P_1 \refs A_1 & P_2 \refs A_2}
\quad
\infer{I \refs I}{}
\qquad
\infer[\mul]{P_1 \mul P_2 \seq{c_1 \mul c_2} Q_1 \mul Q_2}
{P_1 \seq{c_1} Q_1 & P_2 \seq{c_2} Q_2}
\quad
\infer[I]{I \seq{I} I}{}
$$
(we are overloading notation for the monoidal structure on $\cc{D}$ and $\cc{T}$) 
while a closed monoidal refinement system admits the following additional rules:
$$
\small
\infer{\resL[R]{P} \refs \resL[C]{A}}{P \refs A & R \refs C}
\qquad
\infer{\resR[R]{Q} \refs \resR[C]{B}}{R \refs C & Q \refs B}
$$
$$
\small
\infer[\plugL]{P \mul (\resL[R]{P}) \seq{\plugL} R}{}
\quad
\infer[\lambda]{Q \seq{\lc{m}} \resL[R]{P}}{P \mul Q \seq{m} R}
\qquad
\infer[\plugR]{(\resR[R]{Q}) \mul Q \seq{\plugR} R}{}
\quad
\infer[\rho]{P \seq{\rc{m}} \resR[R]{Q}}{P \mul Q \seq{m} R}
$$
Moreover, derivations built using these typing rules satisfy a few equations, which we elide here.
Finally, we remark that the notion of a closed monoidal refinement system can be generalized by allowing the residuation operations to be partial, i.e., by weakening the requirement that $\cc{D}$ and $\cc{T}$ be closed, while maintaining the requirement that the functor $\rs : \cc{D} \to \cc{T}$ preserves any residuals which may exist in $\cc{D}$.
We call such a functor a \definand{logical refinement system}.
Whenever the corresponding residuals exist, a logical refinement system can be treated in essentially the same way as a closed monoidal refinement system, and in particular all of the above rules are valid.


\subsection{Morphisms of refinement systems}
\label{sec:morphisms}

Given a pair of refinement systems $\rs : \cc{D} \to \cc{T}$ and $\b : \cc{E} \to \cc{B}$, by a \definand{morphism of refinement systems} from $\rs$ to $\b$ we mean a pair $F = (F_{\cc{D}},F_{\cc{T}})$ of functors $F_{\cc{D}} : \cc{D} \to \cc{E}$ and $F_{\cc{T}} : \cc{T} \to \cc{B}$ such that the square
$$
\xymatrix{\cc{D} \ar[r]^{F_{\cc{D}}}\ar[d]_\rs & \cc{E} \ar[d]^\b \\ \cc{T} \ar[r]_{F_{\cc{T}}} & \cc{B}}
$$
commutes strictly.
Omitting subscripts on the functors $F$, a morphism from $\rs$ to $\b$ thus induces a pair of rules
$$
\infer{F[P] \refs F[A]}{P \refs A}
\qquad
\infer[F]{F[P] \seq{F[c]} F[Q]}{P \seq{c} Q}
$$
transporting $\rs$-refinements to $\b$-refinements and derivations of $\rs$-judgments to derivations of $\b$-judgments.

Given a pair of morphisms of refinement systems $F = (F_{\cc{D}},F_{\cc{T}}) : \rs \to \b$ and $G = (G_{\cc{D}},G_{\cc{T}}) : \b \to \rs$, an \definand{adjunction of refinement systems} $F \dashv G$ consists of a pair of adjunctions of categories $F_{\cc{D}} \dashv G_{\cc{D}}$ and $F_{\cc{T}} \dashv G_{\cc{T}}$ such that the unit and counit of the adjunction $F_{\cc{D}} \dashv G_{\cc{D}}$ are mapped by $\rs$ and $\b$ onto the unit and counit of $F_{\cc{T}} \dashv G_{\cc{T}}$.
$$
\xymatrixcolsep{3pc}
\xymatrixrowsep{2.5pc}
\xymatrix{\cc{D} \rtwocell^{F_{\cc{D}}}_{G_{\cc{D}}}{'{\bot}}\ar[d]_\rs & \cc{E} \ar[d]^\b \\
 \cc{T} \rtwocell^{F_{\cc{T}}}_{G_{\cc{T}}}{'{\bot}} & \cc{B}}
$$
Writing $\unit$ and $\coun$ (without subscripts) for the unit and counit of both adjunctions $F_{\cc{D}} \dashv G_{\cc{D}}$ and $F_{\cc{T}} \dashv G_{\cc{T}}$, an adjunction of refinement systems thus induces a pair of typing rules
$$
\infer[\unit]{P \seq{\unit} GF[P]}{}\qquad
\infer[\coun]{FG[Q] \seq{\coun} Q}{}
$$
in addition to the typing rules $F$ and $G$, and we remark moreover that typing derivations constructed using these four rules are subject to various equations implied by the definition of an adjunction of categories, such as the triangle laws.

Finally, we say that a morphism of refinement systems $F : \rs \to \b$ is \definand{fully faithful} if the induced typing rule 
$$
\infer[F]{F[P] \seq{F[c]} F[Q]}{P \seq{c} Q}
$$
is \emph{invertible}, in the sense that to any $\b$-derivation
$$
\deduce{F[P] \seq{F[c]} F[Q]}{\beta}
$$
there is a unique $\rs$-derivation
$$
\deduce{P \seq{c} Q}{\pull{F}\beta}
$$
such that
$$
\deduce{F[P] \seq{F[c]} F[Q]}{\beta} \quad=\quad
\infer[F]{F[P] \seq{F[c]} F[Q]}{
\deduce{P \seq{c} Q}{\pull{F}\beta}}.
$$
Equivalently, a morphism of refinement systems $F : \rs \to \b$ 
is fully faithful when the induced functor $\cc{D}\to\cc{E}\times_{\cc{B}}\cc{T}$ to the pullback of $F_{\cc{T}}$ and $\b$ is fully faithful in
the traditional categorical sense.

\subsection{Right adjoints preserve pullbacks}
\label{sec:rapp}

We begin by establishing a basic result about adjunctions of refinement systems, analogous to the well-known fact that in an adjunction of categories the right adjoint functor preserves limits.
Although a similar lemma was proved in \cite{fumex-et-al} for the specific case of two fibrations over the same base, we do not know whether this elementary observation on general functors has already appeared somewhere in the literature. 
In particular, it is worth pointing out that \Cref{prop:rapp} becomes vacuous if stated in terms of traditional \emph{fibred adjunctions} \cite{borceux2,hermida93thesis}, since \emph{any} morphism in the usual 2-category of fibrations must (by definition) preserve cartesian arrows.
\begin{proposition}
\label{prop:rapp}
If $G : \b \to \rs$ is a right adjoint, then $G$ sends $\b$-pullbacks to $\rs$-pullbacks, i.e., for all $c : A \to B$ and $Q \refs B$, whenever the $\b$-pullback $\pull{c}Q$ exists, then the $\rs$-pullback $\pull{G[c]}G[Q]$ exists, and moreover we have that $G[\pull{c}Q] \viso \pull{G[c]}G[Q]$.
\end{proposition}
\begin{proof}
We need to show that $G[\pull{c}Q]$ is a pullback of $G[Q]$ along $G[c]$.
By definition, this means constructing a pair of typing rules
$$
\infer[L\pull{G[c]}]{G[\pull{c}Q] \seq{G[c]} G[Q]}{}
\qquad
\infer[R\pull{G[c]}]{P \seq{d} G[\pull{c}Q]}{P \seq{d;G[c]} G[Q]}
$$
satisfying the $\beta$ and $\eta$ equations.
The left-rule is derived immediately from the left-rule for $\pull{c}Q$ by applying $G$:
$$
\infer[G]{G[\pull{c}Q] \seq{G[c]} G[Q]}{\infer[L\pull{c}]{\pull{c}Q \seq{c} Q}{}}
$$
The right-rule can be derived in a few more steps from the right-rule for $\pull{c}Q$, assuming the existence of an $F$ such that $F \dashv G$:
$$
\infer[\conv_2]{P \seq{d} G[\pull{c}Q]}{
\infer[;]{P \seq{\unit\hid{_X};GF[d];G[\coun\hid{_A}]} G[\pull{c}Q]}{
 \infer[\unit]{P \seq{\unit\hid{_X}} GF[P]}{} &
 \infer[G]{GF[P] \seq{GF[d];G[\coun\hid{_A}]} G[\pull{c}Q]}{
 \infer[R\pull{c}]{F[P] \seq{F[d];\coun\hid{_A}} \pull{c}Q}{
 \infer[\conv_1]{F[P] \seq{F[d];\coun\hid{_A};c} Q}{
 \infer[;]{F[P] \seq{F[d];FG[c];\coun\hid{_B}} Q}{
  \infer[F]{F[P] \seq{F[d];FG[c]} FG[Q]}{P \seq{d;G[c]} G[Q]} &
  \infer[\coun]{FG[Q] \seq{\coun\hid{_B}} Q}{}
}}}}}}
$$
Here at $\conv_1$ and $\conv_2$ we invoke, respectively, naturality of the counit and a triangle law for $F_{\cc{T}} \dashv G_{\cc{T}}$.
Finally, the fact that these typing rules satisfy the $\beta$ and $\eta$ equations can be verified by a long but mechanical calculation (see \Cref{appendix}).
\end{proof}
\noindent
By duality, we also immediately obtain the following:
\begin{proposition}\label{prop:lapp}
If $F : \rs \to \b$ is a left adjoint, then $F$ sends $\rs$-pushforwards to $\b$-pushforwards, i.e., for all $c : A \to B$ and $P \refs A$, whenever the $\rs$-pushforward $\push{c}P$ exists, then the $\b$-pushforward $\push{F[c]}F[P]$ exists, and moreover we have that $\push{F[c]}F[P] \viso F[\push{c}P]$.
\end{proposition}

In passing, we note that \Cref{prop:rapp,prop:lapp} also imply the classical result about ordinary adjunctions of categories.
To see this, begin by observing that for any category $\cc{C}$, one can consider the forgetful functor $\Diagrams{\cc{C}} \to \Cat$ as the \definand{refinement system of diagrams in $\cc{C}$}.
Here $\Diagrams{\cc{C}}$ is the category whose objects are pairs of an indexing category $\cc{I}$ together with a functor $\phi : \cc{I} \to \cc{C}$ (hence, ``diagrams in $\cc{C}$''), and whose morphisms $(\cc{I},\phi) \to (\cc{J},\psi)$ consist of a reindexing functor $F : \cc{I} \to \cc{J}$ together with a natural transformation $\theta : \phi \Rightarrow (F;\psi)$.
The important point is that with respect to this refinement system, a pushforward of a diagram $(\cc{I},\phi)$ along the unique functor $!_{\cc{I}} : \cc{I} \to 1$ corresponds precisely to a colimit of the diagram $\phi : \cc{I} \to \cc{C}$ (and more generally, the pushforward along a functor $F : \cc{I} \to \cc{J}$ corresponds to a left Kan extension of $\phi$ along $F$).
Since an adjunction 
$$
\xymatrix{\cc{C} \rtwocell^L_R{'{\bot}} & \cc{D}}
$$
between two categories $\cc{C}$ and $\cc{D}$ lifts to a \emph{vertical} adjunction
$$
\xymatrixcolsep{1pc}
\xymatrixrowsep{2.5pc}
\xymatrix{\Diagrams{\cc{C}} \rrtwocell^{\Diagrams{L}}_{\Diagrams{R}}{'{\bot}}\ar[dr] && \Diagrams{\cc{D}} \ar[ld] \\
 &\Cat &}
$$
between the respective refinement systems of diagrams, \Cref{prop:lapp} then implies that the left adjoint functor $L$ sends colimits in $\cc{C}$ to colimits in $\cc{D}$ (and more generally, it preserves left Kan extensions in $\Cat$).
By a similar argument, one can use \Cref{prop:rapp} to derive that the right adjoint functor $R$ sends limits in $\cc{D}$ to limits in $\cc{C}$.

Our main application of \Cref{prop:rapp,prop:lapp} in this paper will be the following corollary, about the distributivity properties of pullbacks and pushforwards with respect to tensors and residuals.
\begin{proposition}
\label{prop:logdist}
If $\rs : \cc{D} \to \cc{T}$ is a closed monoidal refinement system, then whenever the corresponding pullbacks and pushforwards exist we have vertical isomorphisms
\begin{align}
\push{(c\mul d)}(P \mul Q) &\viso \push{c}P \mul \push{d}Q \label{eqn:monpf}\tag{a} \\
\resL[\pull{d}R]{\push{c}P} &\viso \pull{(\resL[d]{c})}(\resL[R]{P}) \label{eqn:respp} \tag{b}\\
\resR[\pull{d}R]{\push{c}Q} &\viso \pull{(\resR[d]{c})}(\resR[R]{Q}) \label{eqn:resppr} \tag{c}
\end{align}
\end{proposition}
\begin{proof}
The subtyping judgments in the left-to-right direction are easy to derive, e.g., (\ref{eqn:monpf}) can be derived in any monoidal refinement system:
$$
\small
\infer[L\push{(c\mul d)}]{\push{(c\mul d)}(P \mul Q) \nseq \push{c}P \mul \push{d}Q}{
 \infer[\mul]{P\mul Q \seq{c\mul d} \push{c}P \mul \push{d}Q}{
  \infer[R\push{c}]{P \seq{c} \push{c}P}{} &
  \infer[R\push{d}]{Q \seq{d} \push{d}Q}{}}}
$$
The assumption of monoidal closure implies that these are vertical isomorphisms, using the fact that $\rs$ comes equipped with a family of adjunctions
$$
\xymatrixcolsep{3pc}
\xymatrixrowsep{2.5pc}
\xymatrix{\cc{D} \rtwocell^{P\mul-}_{\resL[-]{P}}{'{\bot}}\ar[d]_\rs & \cc{D} \ar[d]^\rs \\
 \cc{T} \rtwocell^{A\mul-}_{\resL[-]{A}}{'{\bot}} & \cc{T}}
\qquad
\xymatrixcolsep{3pc}
\xymatrixrowsep{2.5pc}
\xymatrix{\cc{D} \rtwocell^{-\mul Q}_{\resR[-]{Q}}{'{\bot}}\ar[d]_\rs & \cc{D} \ar[d]^\rs \\
 \cc{T} \rtwocell^{-\mul B}_{\resR[-]{B}}{'{\bot}} & \cc{T}}
$$
as well as a family of contravariant adjunctions
$$
\xymatrixcolsep{3pc}
\xymatrixrowsep{2.5pc}
\xymatrix{\cc{D} \rtwocell^{\resR[R]{-}}_{\resL[R]{-}}{'{\bot}}\ar[d]_\rs & \cc{D}^\op \ar[d]^{\rs^\op} \\
 \cc{T} \rtwocell^{\resR[C]{-}}_{\resL[C]{-}}{'{\bot}} & \cc{T}^\op}
\qquad
\xymatrixcolsep{3pc}
\xymatrixrowsep{2.5pc}
\xymatrix{\cc{D} \rtwocell^{\resL[R]{-}}_{\resR[R]{-}}{'{\bot}}\ar[d]_\rs & \cc{D}^\op \ar[d]^{\rs^\op} \\
 \cc{T} \rtwocell^{\resL[C]{-}}_{\resR[C]{-}}{'{\bot}} & \cc{T}^\op}
$$
for all $P \refs A$, $Q \refs B$, $R \refs C$.
Explicitly, we have (\ref{eqn:monpf}) by 
\begin{align*}
\push{(c\mul d)}(P \mul Q) 
 &\viso \push{(\id\mul d)}\push{(c\mul \id)}(P \mul Q) \tag{\Cref{prop:pullpushprops}}\\
 &\viso \push{(\id\mul d)}(\push{c}P \mul Q) \tag{$-\mul Q \dashv \resR[-]Q$}\\
 &\viso \push{c}P \mul \push{d}Q \tag{$\push{c}P\mul - \dashv \resL[-]{\push{c}P}$}
\end{align*}
and (\ref{eqn:respp}) (and similarly (\ref{eqn:resppr}))
by
\begin{align*}
\resL[\pull{d}R]{\push{c}P} 
 &\viso \pull{(\resL[d]{\id})} (\resL[R]{\push{c}P}) \tag{$\push{c}P\mul - \dashv \resL[-]{\push{c}P}$} \\
 &\viso \pull{(\resL[d]{\id})} \pull{(\resL[\id]{c})} (\resL[R]{P}) \tag{$\resR[R]{-} \dashv \resL[R]{-}$}\\
 &\viso \pull{(\resL[d]{c})}(\resL[R]{P}) \tag{\Cref{prop:pullpushprops}}
\end{align*}
where in the second-to-last step we use the fact that a $\rs^\op$-pullback is the same thing as a $\rs$-pushforward.
\end{proof}

\section{Representing refinement systems}
\label{sec:representing}

\subsection{The refinement systems of presheaves and of pointed categories}
\label{sec:presheaf}

The \definand{refinement system of presheaves} $\upc : \Psh \to \Cat$ is defined as follows:
\begin{itemize}
\item $\Cat$ is the category whose objects are categories and whose morphisms are functors.
\item Objects of $\Psh$ are pairs $(\cc{A},\phi)$, where $\cc{A}$ is a category and $\phi : \cc{A}^\op \to \Set$ is a contravariant presheaf over that category.
\item Morphisms $(\cc{A},\phi) \to (\cc{B},\psi)$ of $\Psh$ are pairs $(F,\theta)$, where $F : \cc{A} \to \cc{B}$ is a functor and $\theta : \phi \Rightarrow (F;\psi)$ is a natural transformation.
\item $\upc : \Psh \to \Cat$ is the evident forgetful functor.
\end{itemize}
We typically write $\phi \refs \cc{A}$ to indicate that $\phi$ is a presheaf over $\cc{A}$, rather than the more verbose $(\cc{A},\phi) \refs \cc{A}$.
This convention is unproblematic so long as we understand that there is an implicit coercion to view $\phi$ as an object of $\Psh$.
\begin{proposition}
\label{prop:upclogical}
$\upc$ is a closed monoidal refinement system with all pullbacks and pushforwards.
\end{proposition}
\begin{proof}
Pullbacks are defined by precomposition, pushforwards as coends:\footnote{In this paper we adopt the logical notation $\forall x.\Phi(x,x)$ and $\exists y.\Psi(y,y)$ to denote ends and coends, respectively, rather than the more traditional $\int_x \Phi(x,x)$ and $\int^y \Psi(y,y)$ of category theory \cite{kelly}.}
$$
\infer{\pull{F}\psi \refs \cc{A}}{F : \cc{A} \to \cc{B} & \psi \refs \cc{B}}
\qquad \pull{F}\psi \defeq a \mapsto \psi(F a)
$$
$$
\infer{\push{F}\phi \refs \cc{B}}{\phi \refs \cc{A} & F : \cc{A} \to \cc{B}}
\qquad \push{F}\phi \defeq b \mapsto \exists a. \cc{B}(b,Fa) \times \phi(a)
$$
Tensor products and residuals in $\Cat$ are defined using its usual cartesian closed structure (i.e., by building product categories and functor categories), and lifted to $\Psh$ as follows (we show only the definitions of presheaves, not the structural maps):
$$
\infer{I \refs 1}{}\quad
\infer{\phi\mul \psi \refs \cc{A} \times \cc{B}}{\phi \refs \cc{A} & \psi \refs \cc{B}}
\qquad
\phi\mul \psi \defeq (a,b) \mapsto \phi(a) \times \psi(b)
\quad
I \defeq * \mapsto \set{*}
$$
$$
\infer{\resL[\omega]{\phi} \refs [\cc{A},\cc{C}]}{\phi \refs \cc{A} & \omega \refs \cc{C}}
\qquad
\resL[\omega]{\phi} \defeq F  \mapsto \forall a. \phi(a) \to \omega(Fa)
$$
$$
\infer{\resR[\omega]{\psi} \refs [\cc{B},\cc{C}]}{\omega \refs \cc{C} & \psi \refs \cc{B}}
\qquad
\resR[\omega]{\psi} \defeq F  \mapsto \forall a. \psi(a) \to \omega(Fa)
$$
Note that these definitions may be naturally generalized to the refinement system of $\cc{V}$-valued presheaves over $\cc{V}$-enriched categories, but for concreteness we only work with ordinary categories here.
\end{proof}
\noindent
We can also identify a \emph{subsystem} of $\upc$ that will play an important analytical role later on.
Let $\Cat_\bullet$ be the category whose objects consist of categories $\cc{A}$ together with a chosen object $a \in \cc{A}$, and whose morphisms $(\cc{A},a) \to (\cc{B},b)$ are pairs $(F,h)$ consisting of a functor $F : \cc{A} \to \cc{B}$ together with a morphism $h : F(a) \to b$ of $\cc{B}$.
The \definand{refinement system of pointed categories} is defined as the evident forgetful functor $\obc : \Cat_\bullet \to \Cat$.
This is a ``subsystem'' of $\upc$ in the sense that there is a \emph{vertical} morphism of refinement systems $y : \obc \to \upc$, corresponding to the classical Yoneda embedding:
$$
\xymatrix{\Cat_\bullet \ar[d]_\obc\ar[r]^y & \Psh \ar[d]^\upc \\ \Cat \ar@{=}[r] & \Cat}
$$
Read as a vertical morphism of refinement systems, the Yoneda embedding interprets an object $a\in \cc{A}$ as the contravariant presheaf $\cc{A}(-,a)$ over the same category $\cc{A}$.
Finally, since any object $a \in \cc{A}$ can also be seen as a functor $a : 1 \to \cc{A}$ in $\Cat$, let us observe that $y : \obc \to \upc$ may equivalently be defined in terms of pushforward of the unit presheaf:
\begin{proposition}\label{prop:reps}
For all $a \in \cc{A}$, we have $\cc{A}(-,a) \viso \push{a}I$ in $\upc$.
\end{proposition}
\begin{proof}
Immediate from the definition of pushforwards in $\upc$ (see \Cref{prop:upclogical}).
\end{proof}

\subsection{The positive representation of a refinement system}
\label{sec:funrep}

In this section we show that any refinement system $\rs : \cc{D} \to \cc{T}$ has a sound and complete presheaf interpretation, in the sense of a fully faithful morphism of refinement systems $\rs \to \upc$.
We give a direct description of this representation here, as well as some examples, and we will provide some further motivation of the definitions in \Cref{sec:factorpos}.
\begin{definition}For any $\rs$-type $B$, the category $\pt[\rs]{B}$ 
is defined as follows:
\begin{itemize}
\item Objects are pairs $(P \refs A, c : A \to B)$
\item Morphisms $(P_1,c_1) \to (P_2,c_2)$ 
are derivations $\deduce{P_1 \seq{e} P_2}{\alpha}$ such that $c_1 = e;c_2$.
\end{itemize}
\end{definition}
\begin{proposition}\label{prop:slicepos}The assignment $B \mapsto \pt[\rs]{B}$ extends to a functor $\pt[\rs]{(-)} : \cc{T} \to \Cat$.
\end{proposition}
\begin{definition}
For any $\rs$-refinement $Q \refs B$, the presheaf $\pt[\rs]{Q} \refs \pt[\rs]{B}$ is defined on objects by
$$(P,c) \quad \mapsto\quad \set{\alpha \mid \deduce{P \seq{c} Q}{\alpha}}$$
and with the contravariant functorial action transforming any morphism
$(P_1,c_1) \to (P_2,c_2)$
given as a derivation
$$
\deduce{P_1 \seq{e} P_2}{\alpha}
$$
such that $c_1 = e;c_2$ into a typing rule (parametric in $Q$)
$$
\infer{P_1 \seq{c_1} Q}{P_2 \seq{c_2} Q}
$$
derived as
$$\infer[\conv]{P_1 \seq{c_1} Q}{\infer[;]{P_1 \seq{e;c_2} Q}{\deduce{P_1 \seq{e} P_2}{\alpha} & P_2 \seq{c_2} Q}}$$
\end{definition}
\begin{proposition}\label{prop:yonfun}
The assignment $(Q \refs B) \mapsto (\pt[\rs]{B},\pt[\rs]{Q})$ extends to a functor $\pt[\rs]{(-)} : \cc{D} \to \Psh$.
\end{proposition}
\begin{proposition}\label{prop:posembed}
The pair of functors $\pt[\rs]{(-)} : \cc{T} \to \Cat$ and $\pt[\rs]{(-)} : \cc{D} \to \Psh$ define a morphism of refinement systems from $\rs$ to $\upc$, i.e., a commuting square
$$
\xymatrix{\cc{D} \ar[r]^{\pt[\rs]{(-)}}\ar[d]_\rs & \Psh \ar[d]^\upc \\ \cc{T} \ar[r]_{\pt[\rs]{(-)}} & \Cat}
$$
\end{proposition}
\noindent
As we discussed in \Cref{sec:morphisms}, any morphism of refinement systems induces a pair of refinement rules and typing rules, and in this case in particular we have rules
$$
\infer{\pt[\rs]{P} \refs \pt[\rs]{A}}{P \refs A}
\qquad
\infer[\ptfun{\rs}]{\pt[\rs]{P} \seq{\pt[\rs]{c}} \pt[\rs]{Q}}{P \seq{c} Q}
$$
which we call the \definand{positive representation} of the refinement system $\rs$ in the refinement system of presheaves.
\begin{proposition}\label{prop:posff}
The positive representation of $\rs$ is sound and complete, in the sense that the morphism of refinement systems $\pt[\rs]{(-)} : \rs \to \upc$ is fully faithful.
\end{proposition}
\noindent 
Remember that we say a morphism of refinement systems is fully faithful when the induced typing rule is invertible, in this case meaning that to any natural transformation
$$
\deduce{\pt[\rs]{P} \seq{\pt[\rs]{c}} \pt[\rs]{Q}}{\theta}
$$
in $\upc$ there exists a unique $\rs$-derivation
$$
\deduce{P \seq{c} Q}{\pull{(\ptfun{\rs})}\theta}
$$
such that
$$
\deduce{\pt[\rs]{P} \seq{\pt[\rs]{c}} \pt[\rs]{Q}}{\theta} \quad=\quad
\infer[\ptfun{\rs}]{\pt[\rs]{P} \seq{\pt[\rs]{c}} \pt[\rs]{Q}}{
\deduce{P \seq{c} Q}{\pull{(\ptfun{\rs})}\theta}}
$$
To prove \Cref{prop:posff}, we first observe that the presheaves $\pt[\rs]{P}$ are representable (in the classical sense \cite[III.2]{maclane}), so that the positive representation factors via the refinement system of pointed categories.
\begin{proposition}\label{prop:pyonrep}
The morphism $\pt[\rs]{(-)} : \rs \to \upc$ factors as a morphism $\pt[\rs]{(-)} : \rs \to \obc$ followed by the Yoneda embedding,
$$
\vcenter{\xymatrix{\cc{D} \ar[d]_\rs\ar[r]^{\pt[\rs]{(-)}} & \Psh\ar[d]^\upc \\ \cc{T} \ar[r]_{\pt[\rs]{(-)}} & \Cat}}
\quad=\quad
\vcenter{\xymatrix{\cc{D} \ar[d]_\rs\ar[r]^{\pt[\rs]{(-)}} & \Cat_\bullet\ar[d]^\obc\ar[r]^y & \Psh\ar[d]^\upc \\ \cc{T} \ar[r]_{\pt[\rs]{(-)}} & \Cat \ar@{=}[r] & \Cat}}
$$
where $\pt[\rs]{(-)} : \cc{D} \to \Cat_\bullet$ is defined by $\pt[\rs]{P} \defeq (P,\id_A) \in \pt[\rs]{A}$ for all $P \refs A$.
\end{proposition}
\begin{proof}
Immediate from the definitions.
(We overload the notation $\pt[\rs]{P}$ to stand both for the presheaf on $\pt[\rs]{A}$ and for its representing object, but this is harmless since the aspect of $\pt[\rs]{P}$ we are referring to will always be deducible from context.)
\end{proof}
\begin{proof}[Proof of \Cref{prop:posff}]
By \Cref{prop:pyonrep}, it suffices to show separately that each factor $\pt[\rs]{(-)} : \rs \to \obc$ and $y : \obc \to \upc$ is a fully faithful morphism of refinement systems:
\begin{description}
\item[{($\pt[\rs]{(-)} : \rs \to \obc$ is fully faithful.)}] 
Suppose given a derivation of $\pt[\rs]{P} \seq{\pt[\rs]{c}} \pt[\rs]{Q}$ in $\obc$.
By definition of the refinement system $\obc : \Cat_\bullet \to \Cat$ and of the functor $\pt[\rs]{c} : \pt[\rs]{A} \to \pt[\rs]{B}$, such a derivation is the same thing as a morphism $(P,c) \to (Q,\id_B)$ in $\pt[\rs]{B}$.
In turn, it is easy to check from the definition of the category $\pt[\rs]{B}$ that such a morphism is nothing but a $\rs$-derivation of $P \seq{c} Q$.
\item[{($y : \obc \to \upc$ is fully faithful.)}] 
This may of course be reduced to the usual Yoneda lemma, but we can also establish it directly by using the characterization (\Cref{prop:reps}) of representable presheaves as pushforwards of the unit presheaf $\cc{A}(-,a) \viso \push{a}I$.
Consider a $\obc$-typing judgment $a \seq{F} b$ given by a pair of objects $a \in \cc{A}$ and $b \in \cc{B}$ together with a functor $F : \cc{A} \to \cc{B}$.
By the universal property of the pushforward, $\upc$-derivations of $\push{a}I \seq{F} \push{b}I$ are in bijective correspondence with $\upc$-derivations of $I \seq{a;F} \push{b}I$.
The latter correspond exactly to elements of $\cc{B}(F(a),b)$, which by definition are the same thing as derivations of $a \seq{F} b$ in $\obc$.
\end{description}
\end{proof}

\begin{example}
\label{ex:funhoare}
Recall from the Introduction that Hoare Logic can be viewed as a refinement system over a category with one object $W$. 
With respect to this refinement system, the category $\pt{W}$ has pairs $(P,c)$ of a state predicate $P$ and a command $c$ as objects, while a morphism
$$(P_1,c_1) \to (P_2,c_2)$$
corresponds to a derivation of a triple $\triple{P_1}{e}{P_2}$ for some $e$ such that $c_1 = e;c_2$.
Traditionally Hoare Logic is seen through a ``proof irrelevant'' lens, so that a Hoare triple is either valid or invalid, and not much attention is paid to the derivation itself.
If we adopt that simplifying assumption, then the positive embedding is essentially just a set of guarded commands:
$$\pt{Q} = \set{(P,c) \mid \ \vd \triple{P}{c}{Q}}$$
In other words, $(P,c) \in \pt{Q}$ just in case it is possible to run $c$ in a state satisfying the precondition $P$ to obtain a state satisfying $Q$.
\end{example}
\begin{example}
\label{ex:funlinear}
We will formulate the example of sequent calculus for linear logic in a bit more abstract terms as follows.
To any multicategory $\linFm$, there is associated a \emph{free monoidal category} $\FMon{\linFm}$, whose objects (and morphisms) are lists of objects (and morphisms) of $\linFm$, and where the monoidal structure on $\FMon{\linFm}$ is given by concatenation.
Moreover this category is equipped with a forgetful functor
$|{-}| : \FMon{\linFm} \to \simplex$
into the (augmented) simplex category (whose objects are finite ordinals and monotone maps), which interprets a list of objects by its length, and a list of morphisms as a monotone function.

Similarly, there is an analogous construction of the free symmetric monoidal category $\FSMon{\linFm}$ on a symmetric multicategory $\linFm$, where one simply replaces lists by multisets, and the forgetful functor
$$|{-}| : \FSMon{\linFm} \to \Fin$$
maps into the category of finite sets and functions (which contains $\simplex$ as a subcategory).

To model intuitionistic linear sequent calculus along the lines suggested in the Introduction, we will assume given a symmetric multicategory $\linFm$ whose objects are linear logic formulas, and whose multimorphisms are sequent calculus proofs.
Then, we take the category of contexts to be $\linCtx = \FSMon{\linFm}$ and consider the forgetful functor $|{-}| : \linCtx \to \Fin$ as a refinement system.
Since the one-point set $\fin{1}$ is a terminal object in $\Fin$, the category $\pt{\fin{1}}$ is equivalent to $\linCtx$, and the positive embedding of a linear logic formula $C$ (seen as a singleton context $C \refs \fin{1}$) is the presheaf $\pt{C}$ on $\linCtx$ defined exactly as in the Introduction (here we write $\linFm(\G;C)$ for the set of multimorphisms from $\G$ to $C$):
$$
\pt{C} = \G \mapsto \linCtx(\G,C) = \linFm(\G;C) = \set{\pi \mid \deduce{\G \vdash C}{\pi}}
$$\qed
\end{example}

\subsection{Factorization via the free opfibration}
\label{sec:factorpos}

In \Cref{prop:pyonrep} we explained that the positive representation $\pt[\rs]{(-)} : \rs \to \upc$ can be factored as a morphism $\pt[\rs]{(-)} : \rs \to \obc$ followed by the Yoneda embedding $y : \obc \to \upc$ of pointed categories into presheaves.
For the interested (and categorically-minded) reader, in this section we provide some further discussion and motivation of the embedding into pointed categories, explaining its relationship to the \emph{free opfibration} over a functor, as well as the role played by $\obc : \Cat_\bullet \to \Cat$ as a ``universal'' opfibration.

We begin by remarking that the category $\pt[\rs]{B}$ can be seen as an analogue of the slice category over $B$, and reduces to the ordinary slice of $\cc{T}$ over $B$ in the case where $\rs = \id_{\cc{T}} : \cc{T} \to \cc{T}$.
As such, we will sometimes refer to $\pt[\rs]{B}$ as the \definand{$\rs$-slice} (or ``relative slice'') of $B$.
Note that the relative slice construction also appears in \cite[\S1.1.2]{maltsiniotis}, where the notation $\cc{D}/B$ is used instead of $\pt[\rs]{B}$.
\begin{remark}
\label{rem:pyoneda}
In the case where $\rs ={} !_{\cc{D}} : \cc{D} \to 1$, the relative slice over the unique object $*$ of $1$ is $\cc{D}$ itself, and the positive representation $\pt[\rs]{Q}$ of an object $Q \in \cc{D}$ is just $Q$ itself when viewed as a refinement in $\obc$, or the ordinary Yoneda embedding of $Q$ when viewed as a refinement in $\upc$.
More generally, for any $\rs : \cc{D} \to \cc{T}$ and $\rs$-refinement $Q \refs B$, if $B$ is a terminal object in $\cc{T}$ then $\pt[\rs]{B} \viso \cc{D}$, and $\pt[\rs]{Q}$ is represented by the object $Q$ itself.
\end{remark}
\noindent
The fact that the relative slice functor $\pt[\rs]{(-)} : \cc{T} \to \Cat$ reduces to the ordinary slice functor in the case where $\rs$ is the identity can also be understood in terms of the following decomposition:
$$
\xymatrix{\cc{T} \ar[r]^{\pt[\rs]{(-)}} & \Cat}
\quad=\quad
\xymatrix{
\cc{T} \ar[rr]^-{B \mapsto \cc{T}(\rs-,B)} && [\cc{D}^\op,\Set] \ar[r]^-\int & \Cat}
$$
That is, the relative slice functor factors as the \emph{nerve} of $\rs$ followed by the \emph{category of elements} construction.
Seen as a covariant indexed category encoding an opfibration, this composite is just the \definand{free opfibration on $\rs$}, in the sense that the (covariant) category of elements of $\pt[\rs]{(-)} : \cc{T} \to \Cat$ is the comma category $\commacat{\rs}{\cc{T}}$, which has the property that
\begin{enumerate}
\item The projection functor $\cod_\rs : \commacat{\rs}{\cc{T}} \to \cc{T}$ is an opfibration.
\item There is a (vertical) morphism of refinement systems from $\rs$ to $\cod_\rs$,
$$
\xymatrixcolsep{3pc}
\xymatrix{\cc{D} \ar[d]_\rs\ar[r]^{(\id,\rs)} & \commacat{\rs}{\cc{T}} \ar[d]^{\cod_\rs} \\ \cc{T} \ar@{=}[r] & \cc{T}}
$$
where $(\id,\rs) : \cc{D} \to \commacat{\rs}{\cc{T}}$ is the functor sending any $Q\refs B$ to the object $(Q,\id_B,B)$.
\item Any morphism of refinement systems $F : \rs \to \b$ from $\rs$ into an opfibration $\b$ factors uniquely as a morphism $\tilde{F} : \cod_\t \to \b$ composed with the morphism $(\id,\rs) : \rs \to \cod_\rs$.
\end{enumerate}
Next, we can observe that any opfibration has a representation (what one might call the covariant ``Grothendieck representation'') in the refinement system of pointed categories,
$$
\xymatrixcolsep{3pc}
\xymatrix{\cc{E} \ar[d]_\b\ar[r]^{\pgro{\b}} & \Cat_\bullet \ar[d]^{\obc} \\ \cc{B} \ar[r]_{\partial^+\b} & \Cat}
$$
where $\partial^+\b : \cc{B} \to \Cat$ sends any object $X \in \cc{B}$ to the fiber $\cc{E}_X$ of $\b$ over $X$, while $\pgro{\b} : \cc{E} \to \Cat_\bullet$ coerces any refinement $R \refs X$ (i.e., an object $R \in \cc{E}$ such that $\b(R) = X$) into the corresponding element $R \in \cc{E}_X$ of the fiber category.
Note that it is important that $\b : \cc{E} \to \cc{B}$ be an opfibration in order for these operations to define functors.

By combining these two separate observations, we get a simple factorization of the positive embedding into pointed categories.
\begin{proposition}\label{prop:factorization}
The morphism $\pt[\rs]{(-)} : \rs \to \obc$ factors as the free opfibration on $\rs$ followed by the covariant Grothendieck representation:
$$
\vcenter{\xymatrix{\cc{D} \ar[d]_\rs\ar[r]^{\pt[\rs]{(-)}} & \Cat_\bullet\ar[d]^\obc \\ \cc{T} \ar[r]_{\pt[\rs]{(-)}} & \Cat}}
\quad=\quad
\xymatrixcolsep{3pc}
\vcenter{\xymatrix{\cc{D} \ar[d]_\rs\ar[r]^-{(\id,\rs)} & \commacat{\rs}{\cc{T}} \ar[d]^{\cod_\rs}\ar[r]^{\pgro{\cod_\rs}} & \Cat_\bullet\ar[d]^\obc \\ \cc{T} \ar@{=}[r] & \cc{T} \ar[r]_{\partial^+\cod_\rs} & \Cat}}
$$
\end{proposition}

\subsection{The negative representation}
\label{sec:covariant}

Every functor $\rs : \cc{D} \to \cc{T}$ induces an opposite functor $\rs^\op : \cc{D}^\op \to \cc{T}^\op$, and one can consider the positive representation of $\rs^\op$
$$
\xymatrix{\cc{D}^\op \ar[r]^{\pt[\rs^\op]{(-)}}\ar[d]_{\rs^\op} & \Psh \ar[d]^\upc \\ \cc{T}^\op \ar[r]_{\pt[\rs^\op]{(-)}} & \Cat}
$$
as another \definand{negative representation} of $\rs$.  Letting $\nt[\rs]{(-)} \defeq \pt[\rs^\op]{(-)} : \rs^\op \to \upc$, this means we have rules
$$
\infer{\nt[\rs]{P} \refs \nt[\rs]{A}}{P \refs A}
\qquad
\infer[\ntfun{\rs}]{\nt[\rs]{Q} \seq{\nt[\rs]{c}} \nt[\rs]{P}}{P \seq{c} Q}
$$
giving a fully faithful, contravariant embedding of $\rs$ into $\upc$.

Unravelling the definitions, we can verify that
\begin{itemize}
\item $\nt[\rs]{A}$ is the \emph{opposite} of the category whose objects consist of pairs $(c,Q)$ such that $c:A \to B$ and $Q \refs B$, and whose morphisms $(c_1,Q_1) \to (c_2,Q_2)$
correspond to derivations $\deduce{Q_1 \seq{e} Q_2}{\alpha}$ such that $c_1;e = c_2$.
Dually to $\pt[\rs]{A}$, we can read $(\nt[\rs]{A})^\op$ as the \definand{$\rs$-coslice} (or ``relative coslice'') category out of $A$. 
\item For any $\rs$-refinement $P \refs A$, the presheaf $\nt[\rs]{P} \refs \nt[\rs]{A}$ is defined on objects by
$$
(c,Q) \quad\mapsto\quad \set{\alpha \mid \deduce{P \seq{c} Q}{\alpha}}
$$
and on morphisms by
$$\deduce{Q_1 \seq{e} Q_2}{\alpha} \quad \mapsto \quad \infer[\conv]{P \seq{c_2} Q_2}{\infer[;]{P \seq{c_1;e} Q_2}{P \seq{c_1} Q_1 & \deduce{Q_1 \seq{e} Q_2}{\alpha}}}$$
Note that $\nt[\rs]{P}$ is a contravariant presheaf over $\nt[\rs]{A}$, and thus a covariant presheaf over the $\rs$-coslice category.
\end{itemize}
By a similar line of reasoning as in \Cref{sec:funrep,sec:factorpos}, we can decompose the negative representation $\nt[\rs]{(-)} : \rs^\op \to \upc$ into three separate components.
\begin{proposition}\label{prop:pcoyonrep}
The morphism $\nt[\rs]{(-)} : \rs^\op \to \upc$ factors as a morphism $\nt[\rs]{(-)} : \rs^\op \to \obc$ followed by the Yoneda embedding,
$$
\vcenter{\xymatrix{\cc{D}^\op \ar[d]_{\rs^\op}\ar[r]^{\nt[\rs]{(-)}} & \Psh\ar[d]^\upc \\ \cc{T}^\op \ar[r]_{\nt[\rs]{(-)}} & \Cat}}
\quad=\quad
\vcenter{\xymatrix{\cc{D}^\op \ar[d]_{\rs^\op}\ar[r]^{\nt[\rs]{(-)}} & \Cat_\bullet\ar[d]^\obc\ar[r]^y & \Psh\ar[d]^\upc \\ \cc{T}^\op \ar[r]_{\nt[\rs]{(-)}} & \Cat \ar@{=}[r] & \Cat}}
$$
where $\nt[\rs]{(-)} : \cc{D}^\op \to \Cat_\bullet$ is defined by $\nt[\rs]{P} \defeq (\id_A,P) \in \nt[\rs]{A}$ for all $P \refs A$.
\end{proposition}
\begin{proposition}\label{prop:negfactorization}
The morphism $\nt[\rs]{(-)} : \rs^\op \to \obc$ factors as the free fibration on $\rs$ followed by the contravariant Grothendieck representation:
$$
\vcenter{\xymatrix{\cc{D}^\op \ar[d]_{\rs^\op}\ar[r]^{\nt[\rs]{(-)}} & \Cat_\bullet\ar[d]^\obc \\ \cc{T}^\op \ar[r]_{\nt[\rs]{(-)}} & \Cat}}
\quad=\quad
\xymatrixcolsep{3pc}
\vcenter{\xymatrix{\cc{D}^\op \ar[d]_{\rs^\op}\ar[r]^-{(\rs,\id)} & (\commacat{\cc{T}}{\rs})^\op \ar[d]^{\dom_\rs^\op}\ar[r]^{\ngro{\dom_\rs}} & \Cat_\bullet\ar[d]^\obc \\ \cc{T}^\op \ar@{=}[r] & \cc{T}^\op \ar[r]_{\partial^-\dom_\rs} & \Cat}}
$$
\end{proposition}

\begin{example}
\label{ex:opphoare}
If we again follow the classical tradition of treating a Hoare triple $\triple{P}c{Q}$ as either valid or invalid (with no interesting content to the derivation), then the negative representation of a state predicate
$$
\nt{P} = \set{(c,Q) \mid \ \vd \triple{P}c{Q}}
$$
is essentially just the set of all possible \emph{continuations} of a state satisfying $P$.
\end{example}
\begin{example}
\label{ex:opplinear}
With respect to the refinement system $|{-}| : \linCtx \to \Fin$ defined in \Cref{ex:funlinear}, the relative coslice out of $\fin{1}$ has objects $(i : \fin{1} \to \fin{n}, \Gamma \refs \fin{n})$ corresponding to \emph{pointed contexts}, in the sense that the map $i : \fin{1} \to \fin{n}$ serves to select a distinguished formula $A_i$ in $\Gamma = A_1,\dots,A_n$.
A morphism of pointed contexts  $(j,\Delta) \to (i,\Gamma)$ (by which we mean a morphism $(i,\Gamma) \to (j,\Delta)$ in $\nt{\fin{1}}$) corresponds to a linear substitution $\sigma : \Delta \to \Gamma$ whose underlying function maps $j$ to $i$, implying that the chosen formula $B_j$ is used (possibly together with other formulas of $\Delta$) as part of the proof of $A_i$.

To better understand this category, it is helpful to adopt a more evocative notation for pointed contexts.
For example, we could draw the diagram
$$
\xymatrixcolsep{10pt}
\xymatrixrowsep{0pt}
\xymatrix{A_1 & A_2 & A_3 & A_4 \\ \bullet & \bullet & \odot & \bullet}
$$
to represent the pointed context $(i,\Gamma)$ where $\Gamma = A_1,\dots,A_4$ and $i = 3$.
Any morphism of pointed contexts
$$
\xymatrixcolsep{10pt}
\xymatrixrowsep{0pt}
\vcenter{\xymatrix{B_1 & B_2 & B_3 & B_4 \\
\bullet & \odot & \bullet & \bullet}}
\quad\longrightarrow\quad
\vcenter{\xymatrix{A_1 & A_2 & A_3 & A_4 \\
\bullet & \bullet & \odot & \bullet}}
$$
must have an underlying function mapping $2$ to $3$, for example like so:
$$
\xymatrixrowsep{0.5pc}
\vcenter{\xymatrix{
1\ \bullet\ar[r] & \bullet\ 1 \\
2\ \odot\ar[rd] & \bullet\ 2 \\
3\ \bullet\ar[r] & \odot\ 3 \\
4\ \bullet\ar[ruu] & \bullet\ 4}}
$$
In particular, a linear substitution constructed over this specific underlying function consists of a collection of four proofs of the form
$$
\deduce{B_1 \vdash A_1}{\pi_1},\qquad
\deduce{B_4 \vdash A_2}{\pi_2},\qquad
\deduce{B_2,B_3 \vdash A_3}{\pi_3},\quad\text{and}\quad
\deduce{\cdot \vdash A_4}{\pi_4}.
$$
Now, suppose given a formula $A \refs \fin{1}$.
Its negative representation $\nt{A} \refs \nt{\fin{1}}$ corresponds to the presheaf which sends any pointed context
$$
\xymatrixcolsep{10pt}
\xymatrixrowsep{0pt}
\xymatrix{B_1 & \dots & B_j & \dots & B_m \\ \bullet & \dots & \odot & \dots & \bullet}
$$
to the collection of morphisms of pointed contexts
$$
\xymatrixcolsep{10pt}
\xymatrixrowsep{0pt}
\vcenter{\xymatrix{A \\ \odot}}
\quad\longrightarrow\quad
\vcenter{\xymatrix{B_1 & \dots & B_j & \dots & B_m \\ \bullet & \dots & \odot & \dots & \bullet}}
$$
By definition, such a morphism must contain a proof of $A \vd B_j$ together with \emph{closed} proofs of each of the $B_1,\dots,B_{j-1},B_{j+1},\dots,B_m$.

As a shorthand notation, we can write $\fctx{\Delta}{B}$ to stand for a pointed context with chosen formula $B$ and remaining formulas $\Delta$.
Then the presheaf $\nt{A} \refs \nt{\fin{1}}$ is computed on objects by the following expression:
\begin{align*}
\nt{A} = \fctx{\Delta}{B} \mapsto \ & \linFm(A;B) \times \cc{W}(\cdot,\Delta)
\end{align*}
\end{example}

\subsection{Preservation of pullbacks}
\label{sec:pushpull}

We have seen that any refinement system (i.e., any functor) $\rs : \cc{D} \to \cc{T}$ may be embedded both covariantly and contravariantly into the refinement system of pointed categories,
$$
\xymatrix{\rs \ar[r]^{\pt[\rs]{(-)}} & \obc & \rs^\op\ar[l]_{\nt[\rs]{(-)}}}
$$
and that by composing these morphisms with the Yoneda embedding
$$
\xymatrix{\rs \ar[r]^{\pt[\rs]{(-)}}\ar[rd]_{\pt[\rs]{(-)}} & \obc\ar[d]^y & \rs^\op\ar[l]_{\nt[\rs]{(-)}}\ar[ld]^{\nt[\rs]{(-)}} \\ & \upc & }
$$
one obtains two fully faithful presheaf representations of $\rs$.
But why not stop at $\obc$?
As we will see, the benefit of extending the voyage of $\rs$ and $\rs^\op$ all the way into $\upc$ is that this refinement system has a much richer logical structure than $\obc$, which we can apply in order to talk about the original refinement system $\rs$.
By way of illustration, an important property of the positive presheaf representation $\pt[\rs]{(-)} : \rs \to \upc$ is that it \emph{preserves} any pullbacks which may already exist in $\rs$.
\begin{proposition}
\label{prop:pullpres}
Whenever $\pull{c}Q$ exists in $\rs$, we have $\pt[\rs]{(\pull{c}Q)} \viso \pull{(\pt[\rs]{c})}\pt[\rs]{Q}$ in $\upc$.
\end{proposition}
\begin{proof}
By expanding definitions, the elements of $\pt[\rs]{(\pull{c}Q)}$ correspond to $\rs$-derivations
$$P \seqd{\alpha}{d} \pull{c}Q$$
where $P \refs X$ and $d : X \to A$, while the elements of $\pull{(\pt[\rs]{c})}\pt[\rs]{Q}$ correspond to $\rs$-derivations
$$P \seqd{\beta}{d;c} Q.$$
So, the proposition follows from the universal property of the $\rs$-pullback.
\end{proof}
\noindent
As an immediate corollary, we have that the negative representation sends ($\rs$-)pushforwards to ($\upc$-)pullbacks.
\begin{proposition}
\label{prop:pushpres}
Whenever $\push{c}P$ exists in $\rs$, we have $\nt[\rs]{(\push{c}P)} \viso \pull{(\nt[\rs]{c})}\nt[\rs]{P}$ in $\upc$.
\end{proposition}
\noindent
On the other hand, the positive representation need not preserve pushforwards:
although it's true that the $\upc$-subtyping judgment $\push{\pt[\rs]{c}}\pt[\rs]{P} \nseq \pt[\rs]{(\push{c}P)}$ is valid whenever the pushforward $\push{c}P$ exists in $\rs$ (indeed, this is true whenever one has a morphism of refinement systems and the pushforward exists on both sides), in general the converse subtyping judgment need not be valid.
Fortunately, in \Cref{sec:negenc} we will show that although the positive representation need not preserve pushforwards, it at least preserves them ``up to double dualization'' in $\upc$.

\subsection{Preservation of logical connectives up to change-of-basis}
\label{sec:conn}

Suppose that $\rs : \cc{D} \to \cc{T}$ is a monoidal refinement system.
By definition, this means that $\cc{D}$ and $\cc{T}$ are monoidal and that we have a commuting square
$$
\xymatrix{\cc{D}\times \cc{D} \ar[d]_{\mul}\ar[r]^{\rs\times \rs} & \cc{T}\times\cc{T} \ar[d]^{\mul} \\ \cc{D} \ar[r]_-{\rs} & \cc{T}}
$$
(as well as a commuting triangle associated to the tensor unit, but we will ignore the unit in this section, since its treatment is completely analogous).
Since $\upc : \Psh \to \Cat$ is also a monoidal refinement system, the positive representation of $\rs$ thus induces a cube
$$
\xymatrixcolsep{3pc}
\xymatrix@!0{
 & \Psh\times\Psh \ar@{->}[rr]\ar@{->}'[d][dd]
   & & \Cat\times\Cat \ar@{->}[dd]
\\
 \cc{D}\times\cc{D} \ar@{->}[ur]\ar@{->}[rr]\ar@{->}[dd]\drtwocell<\omit>{'a}
 & & \cc{T}\times\cc{T} \ar@{->}[ur]\ar@{->}[dd] \drtwocell<\omit>{'b}
\\
 & \Psh \ar@{->}'[r][rr]
   & & \Cat
\\
 \cc{D} \ar@{->}[rr]\ar@{->}[ur]
 & & \cc{T} \ar@{->}[ur]
}
$$
where all but the left and right faces marked $a$ and $b$ commute strictly.

These latter faces need only commute in the lax sense that there are natural transformations
$$
\xymatrixcolsep{3pc}
\xymatrix{\cc{D}\times\cc{D} \ar[r]^-{\pt[\rs]{(-)}\times\pt[\rs]{(-)}}\ar[d]_{\mul} \drtwocell<\omit>{\str}& \Psh\times\Psh \ar[d]^{\mul} \\ \cc{D} \ar[r]_-{\pt[\rs]{(-)}} & \Psh}
\quad
\xymatrix{\cc{T}\times\cc{T} \ar[r]^-{\pt[\rs]{(-)}\times\pt[\rs]{(-)}}\ar[d]_{\mul} \drtwocell<\omit>{\str}& \Cat\times\Cat \ar[d]^{\times} \\ \cc{T} \ar[r]_-{\pt[\rs]{(-)}} & \Cat}
$$
and moreover the natural transformation on the right is the projection of the one on the left along the cube, this meaning that we have a family of functors
\begin{equation}\label{eqn:str}
\str_{B_1,B_2} : \pt[\rs]{B_1} \times \pt[\rs]{B_2} \to \pt[\rs]{(B_1 \mul B_2)}
\end{equation}
and a family of $\upc$-derivations
\begin{equation}\label{eqn:str2}
\pt[\rs]{Q_1} \mul \pt[\rs]{Q_2} \seqd{\str_{Q_1,Q_2}}{\str_{B_1,B_2}} \pt[\rs]{(Q_1 \mul Q_2)}
\end{equation}
natural in $Q_1\refs B_1$ and $Q_2 \refs B_2$.
Explicitly, the functors $\str_{B_1,B_2}$ are defined by the action sending any pair of objects
$$
(P_1,c_1)\qquad (P_2,c_2)
$$
(where $P_1 \refs A_1, c_1 : A_1 \to B_1, P_2 \refs A_2, c_2 : A_2 \to B_2$)
to the object
$$
(P_1\mul P_2,c_1\mul c_2)
$$
while the natural transformations $m_{Q_1,Q_2}$ are defined by the action sending any pair of $\rs$-derivations
$$
\deduce{P_1 \seq{c_1} Q_1}{\alpha_1}
\qquad
\deduce{P_2 \seq{c_2} Q_2}{\alpha_2}
$$
to the $\rs$-derivation
$$
\infer[\mul]{P_1\mul P_2 \seq{c_1\mul c_2} Q_1 \mul Q_2}{\deduce{P_1 \seq{c_1} Q_1}{\alpha_1} & \deduce{P_2 \seq{c_2} Q_2}{\alpha_2}}
$$
We can summarize all this by saying that the positive representation is a lax morphism of monoidal refinement systems in the expected sense.

From this it follows for purely formal reasons that when $\rs$ is a logical refinement system (i.e., it is monoidal and strictly preserves residuals) we can likewise build functors
\begin{equation}
\costr_{A,C} : \pt[\rs]{(\resL[C]{A})} \to \resL[{\pt[\rs]{C}}]{\pt[\rs]{A}}
\end{equation}
whenever the corresponding residual $\resL[C]{A}$ exists in $\cc{T}$, and $\upc$-derivations
\begin{equation}
\pt[\rs]{(\resL[R]{P})} \seqd{\costr_{P,R}}{\costr_{A,C}} \resL[{\pt[\rs]R}]{\pt[\rs]P}
\end{equation}
whenever the residual $\resL[R]{P}$ exists in $\cc{D}$.
For example, we can define $\costr_{A,C}$ by
$$
\costr_{A,C} \defeq \lc{\str_{A,\resL[C]{A}};\pt[\rs]{\plugL}}
$$
and then construct the derivations as follows:
$$
\small
\infer[\lambda]{\pt[\rs]{(\resL[R]{P})} \seq{\lc{\str;\pt[\rs]{\plugL}}} \resL[{\pt[\rs]{R}}]{\pt[\rs]{P}}}{
\infer[;]{\pt[\rs]P\mul \pt[\rs]{(\resL[R]{P})} \seq{\str;\pt[\rs]{\plugL}} \pt[\rs]{R}}{
\infer[\str]{\pt[\rs]P\mul \pt[\rs]{(\resL[R]{P})} \seq{\str} \pt[\rs]{(P\mul(\resL[R]P))}}{} &
\infer[\ptfun{\rs}]{\pt[\rs]{(P\mul(\resL[R]P))} \seq{\pt[\rs]{\plugL}} \pt[\rs]R}{
\infer[\plugL]{P \mul (\resL[R]{P}) \seq{\plugL} R}{}
}}}
$$
Finally, we also have 
\begin{equation}
\pt[\rs]{(\resR[R]{Q})} \seqd{\costrr_{Q,R}}{\costrr_{B,C}} \resR[{\pt[\rs]R}]{\pt[\rs]Q}
\end{equation}
defined in the analogous way.
Again, all this can be summarized as saying that the positive representation is a lax morphism of logical refinement systems.

However, we can actually establish a much better property about the positive representation, which says that in a certain precise sense it \emph{strongly} preserves the logical structure of $\rs$, but only ``up to change-of-basis''.
\begin{proposition}
\label{prop:genday}
Let $P \refs A$, $Q \refs B$, $R \refs C$ be refinements in a logical refinement system $\rs$.  Then all of the following vertical isomorphisms hold in $\upc$, where (2) and (3) are conditioned on the assumption that the corresponding residuals exist in $\rs$:
\begin{align}
\push{\str_{A,B}}(\pt[\rs]{P} \mul \pt[\rs]{Q}) &\viso \pt[\rs]{(P\mul Q)} \label{eqn:presmul}\tag{a}\\
\pt[\rs]{(\resL[R]{P})} &\viso \pull{\costr_{A,C}}(\resL[{\pt[\rs]{R}}]{\pt[\rs]{P}}) \label{eqn:presresL}\tag{b}\\
\pt[\rs]{(\resR[R]{Q})} &\viso \pull{\costrr_{B,C}}(\resR[{\pt[\rs]{R}}]{\pt[\rs]{Q}}) \label{eqn:presresR}\tag{c}
\end{align}
\end{proposition}
\begin{proof}
We can give a purely formal calculation of (\ref{eqn:presmul}):
\begin{align*}
\push{\str_{A,B}}(\pt[\rs]{P} \mul \pt[\rs]{Q}) 
   &\viso \push{\str_{A,B}}(\push{\pt[\rs]P}I \mul \push{\pt[\rs]Q}I) \tag{\Cref{prop:reps,prop:pyonrep}} 
\\ &\viso \push{\str_{A,B}}\push{(\pt[\rs]P\mul \pt[\rs]Q)}(I \mul I) \tag{\Cref{prop:logdist} }
\\ &\viso \push{\str_{A,B}}\push{(\pt[\rs]P\mul \pt[\rs]Q)}I \tag{$I \viso I\mul I$}
\\ &\viso \push{(\pt[\rs]{P}\mul \pt[\rs]{Q};\str_{A,B})}I \tag{\Cref{prop:pullpushprops} }
\\ &\viso \push{((P,\id_A)\mul (Q,\id_B);\str_{A,B})}I \tag{defn. of $\pt[\rs]P$ and $\pt[\rs]{Q}$}
\\ &\viso \push{(P\mul Q,\id_A\mul \id_B)}I \tag{defn. of $\str_{A,B}$ }
\\ &\viso \push{(P\mul Q,\id_{A\mul B})}I \tag{$\id_{A\mul B} = \id_A\mul \id_B$}
\\ &\viso \push{\pt[\rs]{(P\mul Q)}}I \tag{defn. of $\pt[\rs]{(P\mul Q)}$}
\\ &\viso \pt[\rs]{(P\mul Q)} \tag{\Cref{prop:reps,prop:pyonrep}}
\end{align*}
For (\ref{eqn:presresL}), through similar reasoning we can derive that
\begin{align*}
\pull{\costr_{A,C}}(\resL[{\pt[\rs]{R}}]{\pt[\rs]{P}}) &\viso
 \pull{(\costr_{A,C};(\resL[\id_C]{\pt[\rs]P}))}\pt[\rs]{R} 
\end{align*}
and by expanding definitions we can compute that the elements of the presheaf on the right correspond to derivations
$$
P \mul Q \seqd{\alpha}{\id\mul d} R
$$
where $Q \refs B$ and $d : B \to \resL[C]{A}$.
But then the universal property of the left residual in $\rs$ says that these derivations are in one-to-one correspondence with the elements of $\pt[\rs]{(\resL[R]{P})}$.
The case of (\ref{eqn:presresR}) is similar.
\end{proof}
\noindent
In categorical language, an equivalent way of stating \Cref{prop:genday} is that the morphisms
\begin{align*}
\str_{P,Q} : \pt{P}\mul \pt{Q} &\to \pt{(P\mul Q)}  \\
\costr_{P,R} : \pt{(\resL[R]{P})} &\to \resL[\pt{R}]{\pt{P}} \\
\costrr_{Q,R} : \pt{(\resR[R]{Q})} &\to \resR[\pt{R}]{\pt{Q}} 
\end{align*}
which come from the lax monoidal structure of the functor $\pt[\rs]{(-)} : \cc{D} \to \Psh$ are in fact opcartesian, cartesian, and cartesian, respectively, relative to the functor $\upc : \Psh \to \Cat$.
As a consequence, \Cref{prop:genday} is really an analogue of Day's embedding theorem for monoidal categories \cite{day70}, generalized to the case of logical refinement systems.
\begin{remark}
\label{rem:day}
Consider the case where $\cc{D}$ is a monoidal category and $\rs ={} !_{\cc{D}} : \cc{D} \to 1$ (cf.~\Cref{rem:pyoneda}).
In this case there is a single functor $\str_{*,*} : \cc{D} \times \cc{D} \to \cc{D}$ corresponding to the tensor product on $\cc{D}$, and we get a type-theoretic decomposition of the ``Day construction'' \cite{day70,kelly}, which transports any monoidal category into a closed monoidal category.
In particular, the operation of taking an (``external'') tensor product of presheaves and pushing forward along $\str$ defines an (``internal'') monoidal structure on the presheaf category $[\cc{D}^\op,\Set]$, while the operations of taking an (``external'') residual and pulling back along the functors $\costr$ or $\costrr$ (which are the left/right curryings of $m$) places an (``internal'') closed structure on $[\cc{D}^\op,\Set]$.
Thus, \Cref{prop:genday} specializes to the fact that the Yoneda functor preserves the monoidal structure of $\cc{D}$, as well as any closed structure which may exist.
\end{remark}

Moreover, this remark can be extended to the general case of a monoidal refinement system $\rs : \cc{D}\to\cc{T}$.
We have shown in \cite{mz15popl} that in any such refinement system with enough pushforwards, the fiber~$\cc{D}_W$ of any monoid~$W$ in $\cc{T}$ comes equipped with a monoidal structure defined by pushing forward along the multiplication map $\pairing:W\mul W\to W$:
\begin{align*}
P\otimes_W Q & \defeq \push{\pairing}(P\mul Q)  \tag{$P,Q \refs W$}
\end{align*}
The positive representation induces (by restriction) a functor 
\begin{equation}\label{equation/pt-between-monoidal-fibers}
\pt{(-)} : \cc{D}_W \to [(\pt{W})^\op, \Set] 
\end{equation}
from the fiber of $W$ into the presheaf category over $\pt{W}$.
Moreover, the category $\pt{W}$ inherits a monoidal structure
from the monoid $W$, defined as:
$$
\xymatrixcolsep{3pc}
\xymatrix{
\pt{W}\times\pt{W} \ar[r]^-{\str_{W,W}} & 
\pt{(W\mul W)}\ar[r]^-{\pt{\pairing}} & \pt{W}
}
$$
The associated presheaf category $[(\pt{W})^\op, \Set]$ 
comes thus equipped with a (closed) monoidal structure
provided by the Day tensor product.
Now, the functor~(\ref{equation/pt-between-monoidal-fibers})
is in general only lax monoidal, with the coercion morphism
$$
\xymatrix{\pt{P}\otimes_{\pt{W}} \pt{Q} \ar[r] & \pt{(P\otimes_W Q)}}
$$
constructed by applying the universal property of $\pt{P}\otimes_{\pt{W}} \pt{Q}$ (defined as the $\upc$-pushforward of $\pt{P} \mul \pt{Q}$ along $(m_{W,W};\pt{\pairing})$) to the composite derivation
$$
\pt{P}\mul \pt{Q} \seqd{m_{P,Q}}{m_{W,W}} \pt{(P \mul Q)} \seqd{\pt{\alpha}}{\pt{\pairing}} \pt{(P\otimes_W Q)}
$$
where the right-hand derivation is built by applying the positive representation functor to the derivation
$$P \mul Q \seqd{\alpha}{p} P\otimes_W Q$$
coming from the definition of $P \otimes_W Q$ as a $\rs$-pushforward of $P\mul Q$ along $p$.
\begin{proposition}
\label{prop:laxpres-monoid}
Let $W$ be a monoid in a monoidal refinement system $\rs$ with enough pushforwards.
Then the functor $\pt{(-)} : \cc{D}_W \to [(\pt{W})^\op, \Set]$ is lax monoidal.
In particular, the subtyping judgment $\pt{P}\otimes_{\pt{W}} \pt{Q} \nseq \pt{(P\otimes_W Q)}$ is valid (in $\upc : \Psh \to \Set$) for all $\rs$-refinements $P,Q \refs W$.
\end{proposition}
\noindent
Since the derivation $m_{P,Q}$ is cocartesian (\Cref{prop:genday}), the coercion $\pt{P}\otimes_{\pt{W}} \pt{Q} \to \pt{(P\otimes_W Q)}$ is an isomorphism just in case the positive representation transports the cocartesian derivation $\alpha$ to a cocartesian derivation $\pt{\alpha}$.
This is precisely what happens in the special case discussed in \Cref{rem:day}, where $p$ is equal to the identity.

In the case of a logical refinement system $\rs : \cc{D}\to\cc{T}$ with enough pullbacks, the fiber~$\cc{D}_W$ is not just monoidal, but also closed \cite{mz15popl},
with the residuals defined by pulling back along the left and right curryings of the monoid multiplication map:
\begin{align*}
\impL[W]{R}{P} &\defeq \pull{\lc{p}}{(\resL[R]{P})} \tag{$P,R \refs W$}\\
\impR[W]{R}{Q} &\defeq \pull{\rc{p}}{(\resR[R]{Q})} \tag{$Q,R \refs W$}
\end{align*}
There are two canonical coercion morphisms
$$
\xymatrix{\pt{(\impL[W]{R}{P})} \ar[r] & \impL[\pt{W}]{\pt{R}}{\pt{P}}}
\qquad\text{and}\qquad
\xymatrix{\pt{(\impR[W]{R}{Q})} \ar[r] &\impR[\pt{W}]{\pt{R}}{\pt{Q}}}
$$
induced by the lax monoidal structure of (\ref{equation/pt-between-monoidal-fibers}), 
which can be constructed as the composite derivations
$$
\pt{(\impL[W]{R}P)} \seqd{\pt{\alpha_1}}{\pt{(\lc{p})}} \pt{(\resL[R]P)} \seqd{\costr_{P,R}}{\costr_{W,W}} \resL[{\pt{R}}]{\pt{P}}
$$
and
$$
\pt{(\impR[W]{R}Q)} \seqd{\pt{\alpha_2}}{\pt{(\rc{p})}} \pt{(\resR[R]Q)} \seqd{\costrr_{Q,R}}{\costrr_{W,W}} \resR[{\pt{R}}]{\pt{Q}}
$$
where $\alpha_1$ and $\alpha_2$ are the cartesian derivations coming from the definition of $\impL[W]{R}P$ and $\impR[W]{R}Q$ as pullbacks.
One key difference with the previous situation is that the positive representation preserves all cartesian morphisms (\Cref{prop:pullpres}), which implies that the two coercion morphisms are in fact isomorphisms.
\begin{proposition}
\label{prop:respres}
Let $W$ be a monoid in a logical refinement system $\rs$ with enough residuals and pullbacks.
Then the functor $\pt{(-)} : \cc{D}_W \to [(\pt{W})^\op, \Set]$ preserves residuals.
In particular, we have vertical isomorphisms $\pt{(\impL[W]{R}{P})} \viso \impL[\pt{W}]{\pt{R}}{\pt{P}}$ and $\pt{(\impR[W]{R}{Q})} \viso \impR[\pt{W}]{\pt{R}}{\pt{Q}}$ for all $\rs$-refinements $P,Q,R \refs W$.
\end{proposition}

\section{Duality and negative translation}
\label{sec:duality}

\subsection{Overview}
\label{sec:overview}

The definition of the linear implications $\impL[W]{R}{P}$ and $\impR[W]{R}{Q}$ relative to a monoid $W$ are in fact instances of a more general pattern, which can be implemented in any logical refinement system $\b : \cc{E} \to \cc{B}$ with enough residuals and pullbacks.
Suppose given an arbitrary binary operation 
$$
\pairing : X \mul Y \to Z
$$
in the basis $\cc{B}$.
Then every refinement $R \refs Z$ defines a pair of \emph{dualization operators}
\begin{align*}
\negoL[R]{P} &\defeq \pull{\lc{\pairing}}{(\resL[R]{P})} \tag{$P \refs X$}\\
\negoR[R]{Q} &\defeq \pull{\rc{\pairing}}{(\resR[R]{Q})} \tag{$Q \refs Y$}
\end{align*}
inducing a contravariant adjunction
$$
\xymatrixcolsep{1pc}
\xymatrix{
\cc{E}_X \ar@/^1pc/^{\negoL[R]{(-)}}[rr] & \bot & \cc{E}_Y^\op\ar@/^1pc/^{\negoR[R]{(-)}}[ll]}
$$
between the refinements of $X$ and the refinements of $Y$, as witnessed by the following equivalences of typing and subtyping judgments:
$$
\infer={P \nseq \negoR[R]{Q}}{
\infer={P \seq{\lc{\pairing}} \resR[R]{Q}}{
P \mul Q \seq{\pairing} R}}
\qquad
\infer={Q \nseq \negoL[R]{P}}{
\infer={Q \seq{\rc{\pairing}} \resL[R]{P}}{
P \mul Q \seq{\pairing} R}}
$$
Observe that we don't require that $p$ be the multiplication of a monoid $W$ in order to implement this pattern, although of course we can apply it in that situation.

For example, consider this construction in the refinement system $\upc : \Psh \to \Cat$, applied to a monoidal category $\cc{C}$ seen as an object of $\Cat$ having a tensor product operation $\pairing : \cc{C} \times \cc{C} \to \cc{C}$.
In that case, the fiber associated to $\cc{C}$ is the presheaf category $[\cc{C}^\op,\Set]$, and given a fixed presheaf $R \in [\cc{C}^\op,\Set]$ one recovers a familiar pattern from the theory of linear continuations \cite{thieleckePhD,mellies2012lics}: a contravariant adjunction
$$
\xymatrixcolsep{1pc}
\xymatrix{
[\cc{C}^\op,\Set]\ar@/^1pc/^{\impL{R}{-}}[rr] & \bot & [\cc{C}^\op,\Set]^\op\ar@/^1pc/^{\impR{R}{-}}[ll]}
$$
induced by negation into $R$, where the definition of the two dualization operators coincides with the biclosed monoidal structure on $[\cc{C}^\op,\Set]$ equipped with the Day tensor product (cf.~\Cref{rem:day}).

But besides the connection with linear continuations, the situation is also strongly reminiscent of Isbell duality \cite{isbell} between the categories of covariant and contravariant presheaves over a given category $\cc{C}$.
In that case, however, while still working in the refinement system $\upc : \Psh \to \Cat$, one takes
$$
X = \cc{C} \quad
Y = \cc{C}^\op \quad
Z = \cc{C} \times \cc{C}^\op \qquad
\pairing = \id : \cc{C}\times \cc{C}^\op \to \cc{C}\times \cc{C}^\op
$$
together with $R = \cc{C}(-,-)$ the hom-bimodule of $\cc{C}$.
Then one recovers the contravariant adjunction
$$
\xymatrixcolsep{1pc}
\xymatrix{
[\cc{C}^\op,\Set]\ar@/^1pc/^{\negoL{(-)}}[rr] & \bot & [\cc{C},\Set]^\op\ar@/^1pc/^{\negoR{(-)}}[ll]}
$$
called \emph{Isbell conjugation} \cite[\S7]{lawvere2005seriously}, which transforms any contravariant presheaf into a covariant one, and vice versa.
Expanding the definitions of the refinement type constructors in $\upc : \Psh \to \Cat$ (\Cref{prop:upclogical}), these conjugation operations can be computed explicitly by the following end formulas:
\begin{align*}
\negoL{\phi} &= y \mapsto \forall x. \phi(x) \to \cc{C}(x,y) \\
\negoR{\psi} &= x \mapsto \forall y. \psi(y) \to \cc{C}(x,y)
\end{align*}
One fascinating observation by Isbell is that every pair of representable presheaves 
\begin{align*}
\pt{a} &= \cc{C}(-,a) : \cc{C}^\op \to \Set \\
\nt{a} &= \cc{C}(a,-) : \cc{C} \to \Set
\end{align*}
generated by the same object $a \in \cc{C}$ form a dual pair, in the sense that 
\begin{equation}\label{eqn:duals}
\pt{a} \viso \negoR{(\nt{a})} \qquad\text{and}\qquad
\nt{a} \viso \negoL{(\pt{a})}
\end{equation}
as can be verified by direct application of the Yoneda lemma:
\begin{align*}
\cc{C}(x,a) &\cong \forall y.\cc{C}(a,y) \to \cc{C}(x,y) \\
\cc{C}(a,y) &\cong \forall x.\cc{C}(x,a) \to \cc{C}(x,y)
\end{align*}
Although the equations (\ref{eqn:duals}) may appear counterintuitive if one thinks about the traditional way of working with continuations, the philosophy of Isbell duality says that one can find objects which are invariant with respect to double dualization, provided that the answer type $R$ is sufficiently large and discriminating.

In the specific case of classical Isbell duality, the operation $\pairing$ is trivial, and the role of $R$ is provided by the hom-bimodule.
Our main theorem in this section states that an even more general Isbell-style duality arises for refinement systems, in the sense that any refinement $P \refs A$ in an arbitrary refinement system $\rs$ gives rise to a dual pair 
$$\pt[\rs]{P} \refs \pt[\rs]{A} \qquad
\nt[\rs]{P} \refs \nt[\rs]{A}
$$
in the refinement system of presheaves.
We then develop one application of this theorem, showing how it can be used to explicitly calculate the positive representation of a pushforward, through a sort of negative encoding analogous to the classical double-negation translations of first-order logic into intuitionistic first-order logic.
As a consequence, we also obtain a negative encoding of the positive representation of a fiberwise tensor product, as the double dualization of a Day tensor product.

\subsection{The category of judgments and the presheaf of derivations}

Again, we suppose given an arbitrary refinement system $\rs : \cc{D} \to \cc{T}$.
\begin{definition}
The \definand{category of judgments} $\Typ{\rs}$ is defined as follows:
\begin{itemize}
\item objects are $\rs$-typing judgments: triples $(P,c,Q)$ where $P \refs A$, $c : A \to B$, and $Q \refs B$.
\item morphisms $(P_1,c_1,Q_1) \to (P_2,c_2,Q_2)$ are pairs of $\rs$-derivations
$$
\deduce{P_1 \seq{e} P_2}{\beta} \qquad
\deduce{Q_2 \seq{e'} Q_1}{\gamma}
$$
such that $c_1 = e;c_2;e'$.
\end{itemize}
\end{definition}
\begin{definition}
The \definand{presheaf of derivations} is the refinement $\Der{\rs} \refs \Typ{\rs}$ in $\upc : \Psh \to \Cat$ defined by
$$
\Der{\rs} = (P,c,Q) \mapsto \set{\alpha \mid \deduce{P \seq{c} Q}{\alpha}}
$$
on objects, and with the functorial action transforming any morphism $(P_1,c_1,Q_1) \to (P_2,c_2,Q_2)$ in $\Typ{\rs}$ given as a pair of $\rs$-derivations
$$
  \deduce{P_1 \seq{e} P_2}{\beta} \qquad
  \deduce{Q_2 \seq{e'} Q_1}{\gamma}
$$
such that $c_1 = e;c_2;e'$
into a typing rule
$$
\infer{P_1 \seq{c_1} Q_1}{P_2 \seq{c_2} Q_2}
$$
derived as
$$
\infer[\conv]{P_1 \seq{c_1} Q_1}{
 \infer[;-;]{P_1 \seq{e;c_2;e'} Q_1}{
  \deduce{P_1 \seq{e} P_2}{\beta} &
  P_2 \seq{c_2} Q_2 &
  \deduce{Q_2 \seq{e'} Q_1}{\gamma}
}}
$$
\end{definition}
\noindent
\begin{remark}
\label{rem:ptwisted}
The category of judgments $\Typ{\rs}$ can be seen as an analogue of the ``twisted arrow category'' of $\cc{T}$ \cite{maclane} (see also \cite[p.11]{lawvere70} and \cite[\S1.1.18]{maltsiniotis}), reducing to the opposite of the usual twisted arrow category of $\cc{T}$ in the case $\rs = \id_\cc{T}$.
\end{remark}

\begin{remark}
\label{rem:phom}
In the case where $\rs ={} !_{\cc{D}} : \cc{D} \to 1$, the presheaf of derivations of $\rs$ reduces to the hom-bimodule $\Der{\rs} = \cc{D}(-,-)$ (noting that in that case $\Typ{\rs} \viso \cc{D}\times \cc{D}^\op$).
\end{remark}

\begin{example}
\label{ex:hoarejudge}
For the Hoare Logic refinement system, the category of judgments has objects corresponding to Hoare triples, and has a morphism
$$
\triple{P_1}{c_1}{Q_1} \to \triple{P_2}{c_2}{Q_2}
$$
whenever $c_1$ can be factored as $c_1 = e;c_2;e'$ for some $e$ and $e'$ such that the triples
$$
\triple{P_1}{e}{P_2}\qquad\text{and}\qquad\triple{Q_2}{e'}{Q_1}
$$
are valid.
In particular (in the case where $e$ and $e'$ are equal to the identity), this means that $\Typ{\rs}$ includes morphisms between Hoare triples generated by inverting the ``Rules of Consequence'' \cite{hoare69}, i.e., that there is a morphism
$$
\triple{P_1}{c}{Q_1} \to \triple{P_2}{c}{Q_2}
$$
whenever $\vd P_1 \supset P_2$ and $\vd Q_2 \supset Q_1$.
\end{example}

\subsection{The duality theorem}
\label{sec:conjugation}

We begin by defining a family of \emph{bracket} operations, which will play the role of ``$p$'' in the template described in \Cref{sec:overview}.
\begin{definition}
Let $B$ be a $\rs$-type.
The \definand{$B$-bracket} is the functor $\cut[B] : \pt[\rs]{B} \times \nt[\rs]{B} \to \Typ{\rs}$ defined by $\cut[B]((P,c),(d,R)) = (P,(c;d),R)$.
\end{definition}
\noindent
One way to understand the family of bracket operations is as an \emph{extranatural transformation} from the external product of the relative slice and coslice functors
$$
\xymatrix{\cc{T}\times \cc{T}^\op \ar[rr]^-{\pt[\rs]{(-)}\times \nt[\rs]{(-)}} && \Cat\times \Cat \ar[r]^-{\times} & \Cat}
$$
into the category of judgments, in the sense of
\begin{proposition}\label{prop:extranat}
For any $\rs$-term $c : A \to B$ we have $(\pt[\rs]{c}\times\id_{\nt[\rs]{B}});\cut[B] = (\id_{\pt[\rs]{A}}\times \nt[\rs]{c});\cut[A]$.
\end{proposition}
\noindent
Moreover, although we will not need this fact, the extranatural transformation is universal in the sense that it exhibits the category of judgments as a \emph{coend} $\Typ{\rs} \viso \exists A.\pt[\rs]{A} \times \nt[\rs]{A}$ \cite[see exercise 3 on p. 227 for an analogous remark]{maclane}.

Following the general pattern described in \Cref{sec:overview}, we can use the $B$-bracket in combination with the presheaf of derivations to build a contravariant adjunction
$$
\xymatrixcolsep{1pc}
\xymatrix{
[(\pt[\rs]{B})^\op,\Set]\ar@/^1pc/^{\negoL[R]{(-)}}[rr] & \bot & [(\nt[\rs]{B})^\op,\Set]^\op\ar@/^1pc/^{\negoR[R]{(-)}}[ll]}
$$
between presheaves over $\pt[\rs]B$ and presheaves over $\nt[\rs]B$, where the dualization operators are defined by
\begin{align*}
\negoL{\phi} &\defeq \pull{\lc{\cut[B]}}{(\resL[\Der{\rs}]{\phi})} \tag{$\phi \refs \pt[\rs]{B}$} \\
\negoR{\psi} &\defeq \pull{\rc{\cut[B]}}{(\resR[\Der{\rs}]{\psi})} \tag{$\psi \refs \nt[\rs]{B}$} 
\end{align*}
Moreover, we can establish an Isbell-like duality between the positive and negative representations, relying on the fact that both can be expressed as pullbacks of the presheaf of derivations.
Recall from \Cref{sec:funrep,sec:covariant} that every $\rs$-refinement $Q \refs B$ induces a pair of objects $\pt[\rs]Q \in \pt[\rs]B$ and $\nt[\rs]Q \in \nt[\rs]B$, which represent the corresponding presheaves $\pt[\rs]Q \refs \pt[\rs]B$ and $\nt[\rs]Q \refs \nt[\rs]B$ in $\upc : \Psh \to \Cat$.
Given such a $\rs$-refinement $Q \refs B$, define two functors $k_Q :\pt[\rs]B \to \Typ{\rs}$ and $v_Q : \nt[\rs]B \to \Typ{\rs}$ by
$$
k_Q \defeq (\id_{\pt[\rs]{B}}\times \nt[\rs]{Q});\cut[B]
\quad\text{and}\quad
v_Q \defeq (\pt[\rs]{Q}\times \id_{\nt[\rs]{B}});\cut[B].
$$
\begin{lemma}\label{prop:universal}
For any $\rs$-refinement $Q \refs B$,  we have $\pt[\rs]{Q} \viso \pull{k_Q}\Der{\rs}$ and $\nt[\rs]{Q} \viso \pull{v_Q}\Der{\rs}$.
\end{lemma}
\begin{proof}
Expanding definitions, $k_Q$ and $v_Q$ reduce to the following actions on objects:
\begin{align*}
k_Q &= (P,c) \mapsto (P,c,Q) \\
v_Q &= (d,R) \mapsto (Q,d,R)
\end{align*}
The identities $\pt[\rs]{Q} \viso \pull{k_Q}\Der{\rs}$ and $\nt[\rs]{Q} \viso \pull{v_Q}\Der{\rs}$ are immediate by definition of $\Der{\rs}$.
\end{proof}
\noindent
\begin{theorem}
\label{thm:duality}
For any $\rs$-refinement $Q \refs B$, we have $\nt[\rs]{Q} \viso \negoL{(\pt[\rs]{Q})}$ and $\pt[\rs]{Q} \viso \negoR{(\nt[\rs]{Q})}$.
\end{theorem}
\begin{proof}
The proof is similar to the proof of \Cref{prop:genday}.
We show one case (the other is symmetric):
\begin{align*}
\negoL{(\pt[\rs]{Q})} &\defeq \pull{\lc{\cut[B]}} (\resL[\Der{\rs}]{\pt[\rs]{Q}})
\\ &\viso \pull{\lc{\cut[B]}} (\resL[\Der{\rs}]{\push{\pt[\rs]Q}I}) \tag{\Cref{prop:reps,prop:pyonrep}}
\\ &\viso \pull{\lc{\cut[B]}}\pull{(\resL[\id]{\pt[\rs]Q})} (\resL[\Der{\rs}]{I}) \tag{\Cref{prop:logdist}}
\\ &\viso \pull{\lc{\cut[B]}}\pull{(\resL[\id]{\pt[\rs]Q})} \Der{\rs} \tag{$\Der{\rs} \viso \resL[\Der{\rs}]{I}$}
\\
&\viso \pull{(\lc{\cut[B]};(\resL[\id]{\pt[\rs]Q}))} \Der{\rs} \tag{\Cref{prop:pullpushprops}}
\\ &\viso \pull{((\pt[\rs]Q\times\id);\cut[B])} \Der{\rs} \tag{$\beta$ conversion}
\\ &\viso \nt[\rs]{Q} \tag{\Cref{prop:universal}}
\end{align*}
\end{proof}
\begin{remark}
\label{ex:isbell}
When $\rs ={} !_\cc{D} : \cc{D} \to 1$, the operations $\phi \mapsto \negoL{\phi}$ and $\psi \mapsto \negoR{\psi}$ reduce to Isbell conjugation between the category $[\cc{D}^\op,\Set]$ of contravariant presheaves and the category $[\cc{D},\Set]^\op$ of op'd covariant presheaves, and \Cref{thm:duality} reduces to the fact that Isbell conjugation restricts to an equivalence on representable presheaves.
\end{remark}

\subsection{Negative encodings}
\label{sec:negenc}

We begin by proving a useful lemma.
\begin{lemma}
\label{lem:notpush}
For any $\rs$-term $c : A \to B$ and presheaf $\phi \refs \pt[\rs]A$ we have
$\pull{(\nt[\rs]c)}\negoL{\phi} \viso \negoL{(\push{\pt[\rs]c}\phi)}$.
\end{lemma}
\begin{proof}
The reasoning is similar to the proofs of \Cref{prop:genday} and \Cref{thm:duality}, except for the appeal in the middle to extranaturality of the bracket operations:
\begin{align*}
\negoL{(\push{\pt[\rs]c}\phi)}
 &\viso \pull{(\lc{\cut[B]};(\resL[\id]{\pt[\rs]c}))}(\resL[\Der{\rs}]{\phi})
\\ &\viso \pull{(\pt[\rs]c\mul\id;\cut[B])}(\resL[\Der{\rs}]{\phi})
\\ &\viso \pull{(\id\mul\nt[\rs]c;\cut[A])}(\resL[\Der{\rs}]{\phi}) \tag{\Cref{prop:extranat}}
\\ &\viso \pull{(\nt[\rs]c;\lc{\cut[A]})}(\resL[\Der{\rs}]{\phi})
\\ &\viso \pull{(\nt[\rs]c)}\negoL{\phi}
\end{align*}
\end{proof}
\noindent
A more conceptual way of understanding the lemma is as follows.
Given any term $c : A \to B$ in $\cc{T}$, pulling back and pushing forward along the functors $\pt{c}$ and $\nt{c}$ induces a pair of adjunctions
$$
\xymatrixcolsep{1pc}
\xymatrix{
[(\pt{A})^\op,\Set]\ar@/^1pc/^{\push{\pt{c}}}[rr] & \bot & [(\pt{B})^\op,\Set]\ar@/^1pc/^{\pull{(\pt{c})}}[ll]}
\qquad
\xymatrix{
[(\nt{B})^\op,\Set]\ar@/^1pc/^{\push{\nt{c}}}[rr] & \bot & [(\nt{A})^\op,\Set]\ar@/^1pc/^{\pull{(\nt{c})}}[ll]}
$$
which may be combined with the adjunctions induced by the dualization operators to build a ``thickened square'':
$$
\xymatrixcolsep{1pc}
\xymatrixrowsep{2.5pc}
\xymatrix{
[(\pt{B})^\op,\Set]\ar@/^1pc/^{\pull{(\pt{c})}}[dd]\ar@/^1pc/^{\negoL[R]{(-)}}[rr]
& \bot &
[(\nt{B})^\op,\Set]^\op\ar@/^1pc/^{\push{\nt{c}}^\op}[dd]\ar@/^1pc/^{\negoR[R]{(-)}}[ll] \\
\dashv && \dashv\\
[(\pt{A})^\op,\Set]\ar@/^1pc/^{\push{\pt{c}}}[uu]\ar@/^1pc/^{\negoL[R]{(-)}}[rr] 
& \bot &
[(\nt{A})^\op,\Set]^\op\ar@/^1pc/^{\pull{(\nt{c})}^\op}[uu]\ar@/^1pc/^{\negoR[R]{(-)}}[ll]
}
$$
Beware: not all paths along this diagram commute!
However, \Cref{lem:notpush} says that travelling from the lower left corner to the upper right corner along the outer face is equivalent to travelling with the same origin and destination along the inner face.
Moreover, from the existence of the adjunctions we can automatically derive the following statements, which summarize what happens when one takes different paths along the square.
\begin{corollary}
\label{corr:notpush2}
For any $\rs$-term $c : A \to B$ and presheaves $\psi \refs \nt[\rs]B$, $\rho \refs \pt[\rs]B$, and $\sigma \refs \nt[\rs]A$:
\begin{align}
\pull{(\pt[\rs]c)}\negoR{\psi} &\viso \negoR{(\push{\nt[\rs]c}\psi)} \label{eqn:notpushA}\tag{a} \\
\push{(\nt[\rs]c)}\negoL{\rho} &\nseq \negoL{(\pull{(\pt[\rs]c)}\rho)} \label{eqn:notpushB}\tag{b} \\
\push{(\pt[\rs]c)}\negoR{\sigma} &\nseq \negoR{(\pull{(\nt[\rs]c)}\sigma)} \label{eqn:notpushC}\tag{c}
\end{align}
\end{corollary}
\begin{proof}
(\ref{eqn:notpushA}) follows immediately from \Cref{lem:notpush}, since the two composite functors $\pull{(\pt[\rs]c)} \circ \negoR{(-)}$ and $\negoR{(-)}\circ \push{\nt[\rs]c}$
are right adjoints to the two composite functors $\negoL{(-)}\circ \push{\pt[\rs]{c}}$ and $\pull{(\nt[\rs]c)}\circ \negoL{(-)}$.
Likewise, (\ref{eqn:notpushB}) and (\ref{eqn:notpushC}) follow automatically as mates \cite[\S1.11]{kelly} of the subtyping relations
$$
\pull{(\nt[\rs]c)}\negoL{\phi} \nseq \negoL{(\push{\pt[\rs]c}\phi)}
\quad\text{and}\quad
\pull{(\pt[\rs]c)}\negoR{\psi} \nseq \negoR{(\push{\nt[\rs]c}\psi)}.
$$
Let us nonetheless observe, though, that (\ref{eqn:notpushB}) and (\ref{eqn:notpushC}) are equivalent to the fact that the following typing rules are valid in $\upc : \Psh \to \Cat$:
$$
\infer{\pull{(\pt[\rs]{c})}\rho\mul (\push{\nt[\rs]{c}}\psi) \seq{\cut[A]} \Der{\rs}}{
\rho\mul \psi \seq{\cut[B]} \Der{\rs}}
\qquad
\infer{(\push{\pt[\rs]{c}}\phi)\mul \pull{(\nt[\rs]{c})}\sigma \seq{\cut[B]} \Der{\rs}}{
\phi\mul \sigma \seq{\cut[A]} \Der{\rs}}
$$
The rule on the left, for example, can be derived as follows:
$$
\infer[\mbox{\Cref{prop:logdist}}]{\pull{(\pt{c})}\phi\mul (\push{\nt{c}}\psi) \seq{\cut[A]} \Der{\rs}}{
\infer[L(\id\times \nt{c})]{\push(\id_{\pt{A}}\times\nt{c})(\pull{(\pt{c})}\phi\mul \psi) \seq{\cut[A]} \Der{\rs}}{
\infer[\mbox{\Cref{prop:extranat}}]{\pull{(\pt{c})}\phi\mul \psi \seq{(\id_{\pt{A}}\times\nt{c});\cut[A]} \Der{\rs}}{
\infer[;]{\pull{(\pt{c})}\phi\mul \psi \seq{(\pt{c}\times\id_{\nt{B}});\cut[B]} \Der{\rs}}{
\infer[\mul]{\pull{(\pt{c})}\phi\mul \psi \seq{\pt{c}\times\id_{\nt{B}}} \phi\mul \psi}{
 \infer[L\pull{(\pt{c})}]{\pull{(\pt{c})}\phi \seq{\pt{c}} \phi}{} & \infer[\id]{\psi \seq{\id_{\nt{B}}} \psi}{}}
 & 
\phi\mul \psi \seq{\cut[B]} \Der{\rs}}}}}
$$

\end{proof}
\noindent
As we mentioned at the end of \Cref{sec:pushpull}, the positive representation does not in general preserve pushforwards, although there is always a coercion $\push{\pt{c}}\pt{P} \nseq \pt{(\push{c}P)}$ whenever the pushforward $\push{c}P$ exists in $\rs : \cc{D} \to \cc{T}$.
Similarly, as we discussed at the end of \Cref{sec:conn}, given a monoid $W$ in $\cc{T}$, the induced fiberwise tensor product $\otimes_W$ on $\cc{D}_W$ is not strictly mapped by the functor $\pt{(-)} : \cc{D}_W \to [(\pt{W})^\op,\Set]$ to the Day tensor product $\otimes_{\pt{W}}$, although we have a coercion $\pt{P}\otimes_{\pt{W}}\pt{Q} \nseq \pt{(P\otimes_W Q)}$ for all $\rs$-refinements $P,Q \refs W$ (\Cref{prop:laxpres-monoid}).
One could say that the situation with pullbacks $\pull{c}Q$ and fiberwise residuals $\impL[W]{R}{P}$ and $\impR[W]{R}{Q}$ is nicer, since they are both preserved by the positive representation (\Cref{prop:pullpres,prop:respres}).
However, things are not as bad as they seem for pushforward and fiberwise tensor product, because as we alluded to earlier, this discrepancy may be resolved ``up to double dualization'', by appeal to the Isbell duality theorem for type refinement systems.
\begin{proposition}\label{prop:notpullnot}
Whenever the pushforward $\push{c}P$ exists in $\rs$, we have
\begin{align}
\pt[\rs]{(\push{c}P)} &\viso \negoR{(\pull{(\nt[\rs]c)}\nt[\rs]P)} \label{eqn:notpullnot}\tag{a}\\
\pt[\rs]{(\push{c}P)} &\viso \negoR{(\negoL{(\push{\pt[\rs]c}\pt[\rs]P)})} \label{eqn:notnotpush}\tag{b}
\end{align}
\end{proposition}
\begin{proof}
We can derive equation (\ref{eqn:notpullnot}) in two steps:
$$
\pt[\rs]{(\push{c}P)}  \visod{\text{(\Cref{thm:duality})}} \negoR{\nt[\rs]{(\push{c}P)}} 
  \visod{\text{(\Cref{prop:pushpres})}} \negoR{(\pull{(\nt[\rs]c)}\nt[\rs]P)} 
$$
Then equation (\ref{eqn:notnotpush}) follows in two more steps from (\ref{eqn:notpullnot}):
$$
\negoR{(\pull{(\nt[\rs]c)}\nt[\rs]P)} 
 \visod{\text{(\Cref{thm:duality})}} \negoR{(\pull{(\nt[\rs]c)}\negoL{(\pt[\rs]P)})} 
 \visod{\text{(\Cref{lem:notpush})}} \negoR{(\negoL{(\push{\pt[\rs]c}\pt[\rs]P)})}
$$
\end{proof}
\begin{proposition}\label{prop:notnottensor}
Let $W$ be a monoid in a monoidal refinement system with enough pushforwards.
Then for any $\rs$-refinements $P,Q \refs W$, we have
$
\pt[\rs]{(P \otimes_W Q)} \viso
\negoR{(\negoL{(\pt[\rs]P \otimes_{\pt[\rs]W} \pt[\rs]Q)})}
$.
\end{proposition}
\begin{proof}
We have that
\begin{align*}
\pt[\rs]P \otimes_{\pt{W}} \pt[\rs]Q
  &\defeq \push{(\str_{W,W};\pt{p})}(\pt{P}\mul \pt{Q})
\\ &\viso \push{\pt{p}}\push{\str_{W,W}}(\pt{P}\mul \pt{Q}) \tag{\Cref{prop:pullpushprops}}
\\ &\viso \push{\pt{p}}\pt{(P\mul Q)} \tag{\Cref{prop:genday}}
\end{align*}
and moreover $P \otimes_W Q \defeq \push{p}(P\mul Q)$, so that the proposition follows as a corollary of \Cref{prop:notpullnot}.
\end{proof}
\begin{example}
In Hoare Logic, a pushforward $\push{c}P$ is called a \emph{strongest postcondition} \cite[see \S2]{gordoncollavizza}.
Although in general strongest postconditions need not exist, it is easy to check that in the case when $\push{c}P$ does exist, its positive representation 
$$\pt{(\push{c}P)} = \set{(P',c') \mid \ \vd \triple{P'}{c'}{\push{c}P}}$$
(as computed in \Cref{ex:funhoare}) contains exactly the same guarded commands as 
$$
\negoR{(\pull{(\nt{c})}\nt{P})} = \set{(P',c') \mid \forall (d,R). \ \triple{P}{c;d}{R} \vd \triple{P'}{c';d}{R}}.
$$
Conversely, this latter formula provides a way of reasoning using strongest postconditions, even when they do not exist.
\qed
\end{example}

\begin{example}
Let $A, B \in \linFm$ be two formulas of linear logic, considered as singleton contexts $A \refs \fin{1}$ and $B \refs \fin{1}$  in the refinement system $|{-}| : \linCtx \to \Fin$ of \Cref{ex:funlinear}.
The object $\fin{1}$ is a monoid in $\Fin$, with multiplication $\mu : \fin{2} \to \fin{1}$ defined as the unique map from the two-point set onto the one-point set.
In linear logic, the left introduction rule 
$$
\infer[{\otimes}L]{A\otimes B,\Gamma\vdash C}{A,B,\Gamma \vdash C}
$$
for multiplicative conjunction is invertible, in the sense that it induces a bijection between the proofs of $A,B,\Gamma \vdash C$ and the proofs of $A\otimes B,\Gamma \vdash C$ (considered up to the appropriate equational theory).
Taking $\Gamma$ to be empty, this ensures that the pushforward $\push{\mu}(A,B)$ exists in $|{-}| : \linCtx \to \Fin$, and is given by the formula $A \otimes B \refs \fin{1}$.
Since this pushforward exists for every pair of formulas, by \Cref{prop:laxpres-monoid} there is a lax monoidal functor
$$
\pt{(-)} : \linCtx_{\fin{1}} \to [\linCtx^\op,\Set]
$$
(recall that $\linCtx$ is equivalent to $\pt{\fin{1}}$), with a coercion
$$
\pt{A} \otimes_+ \pt{B} \nseq \pt{(A \otimes B)}
$$
for every $A \refs \fin{1}$ and $B \refs \fin{1}$.
Here we write $\otimes_+$ for the Day tensor product on $[\linCtx^\op,\Set]$, which can be computed as
$$\phi \otimes_+ \psi \defeq \push{(m_{1,1};\pt\mu)}(\phi\mul\psi) \viso \push{\pt{\mu}} \push{m_{1,1}}(\phi\mul\psi)$$
for any pair of presheaves $\phi,\psi\refs \linCtx$, where $m_{1,1} : \linCtx\times\linCtx\to\pt{\fin{2}}$ is the functor defined by the lax monoidal structure of $\pt{(-)} : \Fin \to \Cat$ (see (\ref{eqn:str}) in \Cref{sec:conn}).
In particular, we have that $\pt{A} \otimes_+ \pt{B} \viso \push{\pt{\mu}} \push{m_{1,1}}(\pt{A}\mul\pt{B})$.

Now, an object of $\pt{\fin{2}}$ (namely, a context $\Gamma \refs \fin{n}$ together with a function $f : \fin{n} \to \fin{2}$) is nothing but a partition of a context $\Gamma$ into two disjoint pieces $\Gamma_1$ and $\Gamma_2$, which may be notated conveniently as a shuffle $\Gamma = \Gamma_1\shuffle\Gamma_2$.
So, the functor $m_{1,1} : \linCtx\times\linCtx\to\pt{\fin{2}}$ is the operation which takes a pair of contexts $\Gamma_1$ and $\Gamma_2$ into the corresponding partition $\Gamma = (\Gamma_1,\Gamma_2)$ of a single context into two contiguous pieces.
By \Cref{prop:genday}, we know that $\push{m_{1,1}}(\pt{A}\mul\pt{B}) \viso \pt{(A,B)}$, and the latter simplifies to
$$\pt{(A,B)} = \Gamma_1\shuffle\Gamma_2 \mapsto \linCtx(\Gamma_1,A) \times \linCtx(\Gamma_2,B).$$
Next, consider the pushforward of $\pt{(A,B)}$ along $\pt{\mu} : \pt{\fin{2}} \to \pt{\fin{1}}$.
By the coend formula for pushforwards of presheaves (see \Cref{prop:upclogical}), the presheaf $\push{\pt{\mu}}\pt{(A,B)} \refs \linCtx$ may be calculated as follows:
\begin{align*}
\push{\pt{\mu}}\pt{(A,B)} &= \Gamma \mapsto \exists \Gamma_1,\Gamma_2. \linCtx(\Gamma, (\Gamma_1,\Gamma_2)) \times \linCtx(\Gamma_1,A) \times \linCtx(\Gamma_2,B)
\end{align*}
There is no reason why this presheaf should be isomorphic to
$$\pt{(A\otimes B)} = \Gamma \mapsto \linCtx(\Gamma,A\otimes B).$$
In particular, a counterexample is provided by evaluating both presheaves at the singleton context $\Gamma = A\otimes B$, since one can certainly prove $A\otimes B \vd A \otimes B$, but in general there is no way to split $\Gamma$ into a context proving $A$ and a context proving $B$.

On the other hand, \Cref{prop:notpullnot}(\ref{eqn:notnotpush}) tells us that this mismatch is accounted for by taking a double dual:
\begin{equation}
\pt{(A\otimes B)} \viso \negoR{(\negoL{(\push{\pt{\mu}}\pt{(A,B)})})} \label{eqn:notnot*}
\end{equation}
By \Cref{lem:notpush} and \Cref{thm:duality}, Equation (\ref{eqn:notnot*}) is equivalent to
\begin{equation}
\pt{(A\otimes B)} \viso \negoR{(\pull{(\nt{\mu})}\nt{(A,B)})} \label{eqn:notnot*2}
\end{equation}
which can be derived from the simple equation
\begin{equation}
\nt{(A\otimes B)} \viso \pull{(\nt{\mu})}\nt{(A,B)} \label{eqn:notnot*3}
\end{equation}
by one application of \Cref{thm:duality}.
\Cref{eqn:notnot*3} itself follows from the definition of $A\otimes B$ as a pushforward and \Cref{prop:pushpres}.
In order to understand this equation, recall from \Cref{ex:opplinear} that $\nt{\fin{1}}$ is the category of pointed contexts, and that the negative representation $\nt{(A\otimes B)} \refs \nt{\fin{1}}$ is defined by the action
$$
\fctx{\Delta}{C} \mapsto \linFm(A\otimes B;C) \times \linCtx(\cdot, \Delta).
$$
In other words, $\nt{(A\otimes B)}$ transports a pointed context $\fctx{\Delta}{C}$ into the set of tuples consisting of a proof of $A\otimes B \vdash C$ together with a closed proof of each formula in $\Delta$.
A careful computation shows that $\pull{(\nt{\mu})}\nt{(A,B)}$ is defined by the action
$$\fctx{\Delta}{C} \mapsto \linFm(A,B;C) \times \linCtx(\cdot, \Delta).$$
So \Cref{eqn:notnot*3} reduces to the fact that the left introduction rule ${\otimes}L$ is invertible.
In particular, observe that whereas we could distinguish $\pt{(A\otimes B)}$ from $\push{\pt{\mu}}\pt{(A,B)}$ by considering the context $\Gamma = A\otimes B$, their duals $\nt{(A\otimes B)}$ and $\pull{(\nt{\mu})}\nt{(A,B)}$ cannot be distinguished by \emph{pointed} contexts.


\end{example}



\onecolumn

\appendix
\section{Proof of \Cref{prop:rapp}}
\label{appendix}

We complete the proof of \Cref{prop:rapp} by showing that the typing rules $L\pull{G[c]}$ and $R\pull{G[c]}$ defined by
$$
\infer[G]{G[\pull{c}Q] \seq{G[c]} G[Q]}{\infer[L\pull{c}]{\pull{c}Q \seq{c} Q}{}}
\qquad\qquad
\infer[\conv_2]{P \seq{d} G[\pull{c}Q]}{
\infer[;]{P \seq{\unit\hid{_X};GF[d];G[\coun\hid{_A}]} G[\pull{c}Q]}{
 \infer[\unit]{P \seq{\unit\hid{_X}} GF[P]}{} &
 \infer[G]{GF[P] \seq{GF[d];G[\coun\hid{_A}]} G[\pull{c}Q]}{
 \infer[R\pull{c}]{F[P] \seq{F[d];\coun\hid{_A}} \pull{c}Q}{
 \infer[\conv_1]{F[P] \seq{F[d];\coun\hid{_A};c} Q}{
 \infer[;]{F[P] \seq{F[d];FG[c];\coun\hid{_B}} Q}{
  \infer[F]{F[P] \seq{F[d];FG[c]} FG[Q]}{P \seq{d;G[c]} G[Q]} &
  \infer[\coun]{FG[Q] \seq{\coun\hid{_B}} Q}{}
}}}}}}
$$
satisfy the $\beta$ equation
$$
\small
\infer[;]{P \seq{d;G[c]} G[Q]}{
 \infer[R\pull{G[c]}]{P \seq{d} G[\pull{c}Q]}{P \seqd{\beta}{d;G[c]} G[Q]} &
 \infer[L\pull{G[c]}]{G[\pull{c}Q] \seq{G[c]} G[Q]}{}
}
\ \ \deq\ \ 
P \seqd{\beta}{d;G[c]} G[Q]
$$
as well as the $\eta$ equation
$$
\small
P \seqd{\eta}{d} G[\pull{c}Q]
\ \ \deq\ \ 
\infer[R\pull{G[c]}]{P \seq{d} G[\pull{c}Q]}
{\infer[;]{P \seq{d;G[c]} G[Q]}
{P \seqd{\eta}{d} G[\pull{c}Q] & \infer[L\pull{G[c]}]{G[\pull{c}Q] \seq{G[c]} G[Q]}{}}}
$$
In both cases, we establish these equations by showing a series of typing derivations, with each successive derivation related to its predecessor by an elementary step of reasoning.

\subsection{The $\beta$ equation}

$$
\small
\infer[;]{P \seq{d;G[c]} G[Q]}{
\infer[\conv]{P \seq{d} G[\pull{c}Q]}{
\infer[;]{P \seq{\unit_X;GF[d];G[\coun_A]} G[\pull{c}Q]}{
 \infer[\unit]{P \seq{\unit_X} GF[P]}{} &
 \infer[G]{GF[P] \seq{GF[d];G[\coun_A]} G[\pull{c}Q]}{
 \infer[R\pull{c}]{F[P] \seq{F[d];\coun_A} \pull{c}Q}{
 \infer[\conv]{F[P] \seq{F[d];\coun_A;c} Q}{
 \infer[;]{F[P] \seq{F[d];FG[c];\coun_B} Q}{
  \infer[F]{F[P] \seq{F[d];FG[c]} FG[Q]}{P \seqd{\beta}{d;G[c]} G[Q]} &
  \infer[\coun]{FG[Q] \seq{\coun_B} Q}{}
}}}}}}
&
\infer[G]{G[\pull{c}Q] \seq{G[c]} G[Q]}{\infer[L\pull{c}]{\pull{c}Q \seq{c} Q}{}}
}
$$

\hrule

$$
\small
\infer[\conv]{P \seq{d;G[c]} G[Q]}{
\infer[;]{P \seq{\unit_X;GF[d];G[\coun_A];G[c]} G[Q]}{
\infer[;]{P \seq{\unit_X;GF[d];G[\coun_A]} G[\pull{c}Q]}{
 \infer[\unit]{P \seq{\unit_X} GF[P]}{} &
 \infer[G]{GF[P] \seq{GF[d];G[\coun_A]} G[\pull{c}Q]}{
 \infer[R\pull{c}]{F[P] \seq{F[d];\coun_A} \pull{c}Q}{
 \infer[\conv]{F[P] \seq{F[d];\coun_A;c} Q}{
 \infer[;]{F[P] \seq{F[d];FG[c];\coun_B} Q}{
  \infer[F]{F[P] \seq{F[d];FG[c]} FG[Q]}{P \seqd{\beta}{d;G[c]} G[Q]} &
  \infer[\coun]{FG[Q] \seq{\coun_B} Q}{}
}}}}}
&
\infer[G]{G[\pull{c}Q] \seq{G[c]} G[Q]}{\infer[L\pull{c}]{\pull{c}Q \seq{c} Q}{}}
}}
$$

\hrule

$$
\small
\infer[\conv]{P \seq{d;G[c]} G[Q]}{
\infer[;]{P \seq{\unit_X;GF[d];G[\coun_A];G[c]} G[Q]}{
 \infer[\unit]{P \seq{\unit_X} GF[P]}{} &
 \infer[;]{GF[P] \seq{GF[d];G[\coun_A];G[c]} G[\pull{c}Q]}{
 \infer[G]{GF[P] \seq{GF[d];G[\coun_A]} G[\pull{c}Q]}{
 \infer[R\pull{c}]{F[P] \seq{F[d];\coun_A} \pull{c}Q}{
 \infer[\conv]{F[P] \seq{F[d];\coun_A;c} Q}{
 \infer[;]{F[P] \seq{F[d];FG[c];\coun_B} Q}{
  \infer[F]{F[P] \seq{F[d];FG[c]} FG[Q]}{P \seqd{\beta}{d;G[c]} G[Q]} &
  \infer[\coun]{FG[Q] \seq{\coun_B} Q}{}
}}}}
&
\infer[G]{G[\pull{c}Q] \seq{G[c]} G[Q]}{\infer[L\pull{c}]{\pull{c}Q \seq{c} Q}{}}
}}}
$$

\hrule

$$
\small
\infer[\conv]{P \seq{d;G[c]} G[Q]}{
\infer[;]{P \seq{\unit_X;GF[d];G[\coun_A];G[c]} G[Q]}{
 \infer[\unit]{P \seq{\unit_X} GF[P]}{} &
 \infer[G]{GF[P] \seq{GF[d];G[\coun_A];G[c]} G[\pull{c}Q]}{
 \infer[;]{F[P] \seq{F[d];\coun_A;c} G[\pull{c}Q]}{
 \infer[R\pull{c}]{F[P] \seq{F[d];\coun_A} \pull{c}Q}{
 \infer[\conv]{F[P] \seq{F[d];\coun_A;c} Q}{
 \infer[;]{F[P] \seq{F[d];FG[c];\coun_B} Q}{
  \infer[F]{F[P] \seq{F[d];FG[c]} FG[Q]}{P \seqd{\beta}{d;G[c]} G[Q]} &
  \infer[\coun]{FG[Q] \seq{\coun_B} Q}{}
}}}
&
\infer[L\pull{c}]{\pull{c}Q \seq{c} Q}{}
}}}}
$$

\hrule

$$
\small
\infer[\conv]{P \seq{d;G[c]} G[Q]}{
\infer[;]{P \seq{\unit_X;GF[d];G[\coun_A];G[c]} G[Q]}{
 \infer[\unit]{P \seq{\unit_X} GF[P]}{} &
 \infer[G]{GF[P] \seq{GF[d];G[\coun_A];G[c]} G[\pull{c}Q]}{
 \infer[\conv]{F[P] \seq{F[d];\coun_A;c} Q}{
 \infer[;]{F[P] \seq{F[d];FG[c];\coun_B} Q}{
  \infer[F]{F[P] \seq{F[d];FG[c]} FG[Q]}{P \seqd{\beta}{d;G[c]} G[Q]} &
  \infer[\coun]{FG[Q] \seq{\coun_B} Q}{}
}}}
}}
$$

\hrule

$$
\small
\infer[\conv]{P \seq{d;G[c]} G[Q]}{
\infer[;]{P \seq{\unit_X;GF[d];GFG[c];G[\coun_A]} G[Q]}{
 \infer[\unit]{P \seq{\unit_X} GF[P]}{} &
 \infer[G]{GF[P] \seq{GF[d];GFG[c];G[\coun_B]} G[\pull{c}Q]}{
 \infer[;]{F[P] \seq{F[d];FG[c];\coun_B} Q}{
  \infer[F]{F[P] \seq{F[d];FG[c]} FG[Q]}{P \seqd{\beta}{d;G[c]} G[Q]} &
  \infer[\coun]{FG[Q] \seq{\coun_B} Q}{}
}}
}}
$$

\hrule

$$
P \seqd{\beta}{d;G[c]} G[Q]
$$

\pagebreak
\subsection{The $\eta$ equation}

$$P \seqd{\eta}{d} G[\pull{c}Q]$$

\hrule

$$
\small
\infer[\conv]{P \seq{d} G[\pull{c}Q]}{
\infer[;]{P \seq{\unit_X;GF[d];G[\coun_A]} G[\pull{c}Q]}{
 \infer[\unit]{P \seq{\unit_X} GF[P]}{} &
 \infer[G]{GF[P] \seq{GF[d];G[\coun_A]} G[\pull{c}Q]}{
 \infer[;]{FG[\pull{c}Q] \seq{F[d];\coun_A} G[Q]}{
  \infer[F]{F[P] \seq{F[d]} FG[\pull{c}Q]}{
  P \seqd{\eta}{d} G[\pull{c}Q]} &
  \infer[\coun]{FG[\pull{c}Q] \seq{\coun_A} \pull{c}Q}{}}
}}}
$$

\hrule

$$
\small
\infer[\conv]{P \seq{d} G[\pull{c}Q]}{
\infer[;]{P \seq{\unit_X;GF[d];G[\coun_A]} G[\pull{c}Q]}{
 \infer[\unit]{P \seq{\unit_X} GF[P]}{} &
 \infer[G]{GF[P] \seq{GF[d];G[\coun_A]} G[\pull{c}Q]}{
 \infer[R\pull{c}]{F[P] \seq{F[d];\coun_A} \pull{c}Q}{
 \infer[;]{F[P] \seq{F[d];\coun_A;c} Q}{
 \infer[;]{FG[\pull{c}Q] \seq{F[d];\coun_A} G[Q]}{
  \infer[F]{F[P] \seq{F[d]} FG[\pull{c}Q]}{
  P \seqd{\eta}{d} G[\pull{c}Q]} &
  \infer[\coun]{FG[\pull{c}Q] \seq{\coun_A} \pull{c}Q}{}}
  &
  \infer[L\pull{c}]{\pull{c}Q \seq{c} Q}{}
}}}}}
$$

\hrule

$$
\small
\infer[\conv]{P \seq{d} G[\pull{c}Q]}{
\infer[;]{P \seq{\unit_X;GF[d];G[\coun_A]} G[\pull{c}Q]}{
 \infer[\unit]{P \seq{\unit_X} GF[P]}{} &
 \infer[G]{GF[P] \seq{GF[d];G[\coun_A]} G[\pull{c}Q]}{
 \infer[R\pull{c}]{F[P] \seq{F[d];\coun_A} \pull{c}Q}{
 \infer[;]{F[P] \seq{F[d];\coun_A;c} Q}{
  \infer[F]{F[P] \seq{F[d]} FG[\pull{c}Q]}{
  P \seqd{\eta}{d} G[\pull{c}Q]} &
\infer[;]{FG[\pull{c}Q] \seq{\coun_A;c} G[Q]}{
 \infer[\coun]{FG[\pull{c}Q] \seq{\coun_A} \pull{c}Q}{}
  &
  \infer[L\pull{c}]{\pull{c}Q \seq{c} Q}{}
}}}}}}
$$
\pagebreak

\hrule

$$
\small
\scalebox{0.9}{
\infer[\conv]{P \seq{d} G[\pull{c}Q]}{
\infer[;]{P \seq{\unit_X;GF[d];G[\coun_A]} G[\pull{c}Q]}{
 \infer[\unit]{P \seq{\unit_X} GF[P]}{} &
 \infer[G]{GF[P] \seq{GF[d];G[\coun_A]} G[\pull{c}Q]}{
 \infer[R\pull{c}]{F[P] \seq{F[d];\coun_A} \pull{c}Q}{
 \infer[\conv]{F[P] \seq{F[d];\coun_A;c} Q}{
 \infer[;]{F[P] \seq{F[d];FG[c];\coun_B} Q}{
  \infer[F]{F[P] \seq{F[d]} FG[\pull{c}Q]}{
  P \seqd{\eta}{d} G[\pull{c}Q]} &
\infer[;]{FG[\pull{c}Q] \seq{FG[c];\coun_B} G[Q]}{
 \infer[F]{FG[\pull{c}Q] \seq{FG[c]} FG[Q]}{
 \infer[G]{G[\pull{c}Q] \seq{G[c]} G[Q]}{\infer[L\pull{c}]{\pull{c}Q \seq{c} Q}{}}}
 &
  \infer[\coun]{FG[Q] \seq{\coun_B} Q}{}
}}}}}}}}
$$

\hrule

$$
\small
\scalebox{0.9}{
\infer[\conv]{P \seq{d} G[\pull{c}Q]}{
\infer[;]{P \seq{\unit_X;GF[d];G[\coun_A]} G[\pull{c}Q]}{
 \infer[\unit]{P \seq{\unit_X} GF[P]}{} &
 \infer[G]{GF[P] \seq{GF[d];G[\coun_A]} G[\pull{c}Q]}{
 \infer[R\pull{c}]{F[P] \seq{F[d];\coun_A} \pull{c}Q}{
 \infer[\conv]{F[P] \seq{F[d];\coun_A;c} Q}{
 \infer[;]{F[P] \seq{F[d];FG[c];\coun_B} Q}{
  \infer[;]{F[P] \seq{F[d];FG[c]} FG[Q]}{ 
  \infer[F]{F[P] \seq{F[d]} FG[\pull{c}Q]}{
  P \seqd{\eta}{d} G[\pull{c}Q]} &
 \infer[F]{FG[\pull{c}Q] \seq{FG[c]} FG[Q]}{
 \infer[G]{G[\pull{c}Q] \seq{G[c]} G[Q]}{\infer[L\pull{c}]{\pull{c}Q \seq{c} Q}{}}}} &
  \infer[\coun]{FG[Q] \seq{\coun_B} Q}{}
}}}}}}}
$$

\hrule

$$
\small
\infer[\conv]{P \seq{d} G[\pull{c}Q]}{
\infer[;]{P \seq{\unit_X;GF[d];G[\coun_A]} G[\pull{c}Q]}{
 \infer[\unit]{P \seq{\unit_X} GF[P]}{} &
 \infer[G]{GF[P] \seq{GF[d];G[\coun_A]} G[\pull{c}Q]}{
 \infer[R\pull{c}]{F[P] \seq{F[d];\coun_A} \pull{c}Q}{
 \infer[\conv]{F[P] \seq{F[d];\coun_A;c} Q}{
 \infer[;]{F[P] \seq{F[d];FG[c];\coun_B} Q}{
  \infer[F]{F[P] \seq{F[d];FG[c]} FG[Q]}{ 
  \infer[;]{P \seq{d;G[c]} G[Q]}{
   P \seqd{\eta}{d} G[\pull{c}Q] &
   \infer[G]{G[\pull{c}Q] \seq{G[c]} G[Q]}{\infer[L\pull{c}]{\pull{c}Q \seq{c} Q}{}}}} &
  \infer[\coun]{FG[Q] \seq{\coun_B} Q}{}
}}}}}}
$$


\begin{thebibliography}{}

\bibitem[Andreoli 1992]{andreoli92}
Jean-Marc Andreoli. Logic programming with focussing proofs in Linear Logic. \emph{Journal of Logic and Computation} 2:3 (1992).

\bibitem[Borceux 1994]{borceux2}
Francis Borceux. \emph{Handbook of Categorical Algebra 2: Categories and Structures}. Cambridge University Press, 1994.

\bibitem[Day 1970]{day70}
B.J. Day. On closed categories of functors, {\em Lecture Notes in Mathematics} 137 (1970), 1--38.

\bibitem[Gentzen 1935]{gentzen1935thesis}
Gerhard Gentzen. {\em Untersuchungen über das logische Schliessen (Investigations into Logical Inference)}, Ph.D. thesis, Universität Göttingen. English translation in {\em The Collected Papers of Gerhard Gentzen}, M. Szabo (ed.), Amsterdam: North Holland (1969).

\bibitem[Ghani, Johann, Fumex 2013]{fumex-et-al}
Neil Ghani, Patricia Johann, Cl\'ement Fumex.
Indexed Induction and Coinduction, Fibrationally.
{\em Logical Methods in Computer Science} 9(3:6) (2013), 1--31.

\bibitem[Girard 1987]{girard87linear}
Jean-Yves Girard. Linear logic. {\em Theoretical Computer Science} 50 (1987), 1--102.

\bibitem[Gordon and Collavizza 2010]{gordoncollavizza}
Mike Gordon and Hélène Collavizza.  Forward with Hoare.  In {\em Reflections on the Work of C.A.R. Hoare}, Cliff Jones, A. W. Roscoe, Kenneth R. Wood (eds.). Springer, 2010.

\bibitem[Hermida 1993]{hermida93thesis}
Claudio Hermida. {\em Fibrations, Logical predicates and indeterminates}, Ph.D. thesis, University of Edinburgh, November 1993.

\bibitem[Hoare 1969]{hoare69}
C.A.R. Hoare. An Axiomatic Basis for Computer Programming, {\em Communications of the ACM} 12:10 (1969).

\bibitem[Isbell 1966]{isbell}
John Isbell. Structure of categories, {\em Bulletin of the American Mathematical Society} 72 (1966), 619--655.

\bibitem[Kelly 1982]{kelly}
Max Kelly. {\em Basic concepts in enriched category theory}.  Cambridge University Press, 1982.


\bibitem[Lawvere 1970]{lawvere70}
F. William Lawvere. Equality in hyperdoctrines and comprehension schema as an adjoint functor, In \emph{Proceedings of the AMS Symposium on Pure Mathematics XVII} (1970), 1--14.

\bibitem[Lawvere 2005]{lawvere2005seriously}
F. William Lawvere. Taking Categories Seriously. \emph{Reprints in Theory and Applications of Categories} 8 (2005), 1--24.

\bibitem[Mac Lane 1971]{maclane}
Saunders Mac Lane. {\em Categories for the Working Mathematician}. Springer, 1971.

\bibitem[Maltsiniotis 2005]{maltsiniotis}
Georges Maltsiniotis. {\em La théorie de l'homotopie de Grothendieck}. Astérisque, 2005.

\bibitem[Melliès 2012]{mellies2012lics}
Paul-André Melliès. Game Semantics in String Diagrams. In {\em Proceedings of the 27th Annual IEEE Conference on Logic in Computer Science}, Dubrovnik, 2012.

\bibitem[Melliès and Zeilberger 2015]{mz15popl}
Paul-André Melliès and Noam Zeilberger. Functors are Type Refinement Systems.
In {\em Proceedings of the 42nd {ACM} {SIGPLAN-SIGACT} Symposium on Principles of Programming}, Mumbai, 2015.

\bibitem[Thielecke 1997]{thieleckePhD}
Hayo Thielecke. {\em Categorical Structure of Continuation Passing Style}, Ph.D. thesis, University of Edinburgh, 1997.

\end{thebibliography}
\end{document}